\documentclass[sigplan,10pt]{acmart}
\usepackage{xspace}
\usepackage{threeparttable}
\usepackage{color}
\usepackage{url}
\usepackage{subcaption}
\usepackage{paralist}
\usepackage{wrapfig}
\usepackage{multirow}
\usepackage{listings}
\usepackage{comment}
\usepackage{booktabs}
\usepackage{makecell}
\usepackage{boxedminipage}
\usepackage{amsmath}
\usepackage{graphicx}
\usepackage{minibox}
\usepackage{xcolor}
\usepackage{subfiles}
\usepackage[cachedir=.]{minted}
\usepackage{ulem}
\usepackage{adjustbox}
\usepackage{pifont}
\usepackage{amsthm}
\usepackage{natbib}
\usepackage{algorithm}
\usepackage{float}
\usepackage{multicol}
\usepackage{diagbox}
\usepackage{slashbox}
\usepackage{balance}
\usepackage{algpseudocode}
\usepackage{tikz}
\usetikzlibrary{positioning,arrows.meta,shapes}

\usepackage{xcolor}

\definecolor{purple}{RGB}{105,33,106} 
\definecolor{darkred}{RGB}{154,48,53} 
\definecolor{paramBlue}{RGB}{39,78,105}
\definecolor{darkgreen}{rgb}{0.0, 0.5, 0.0}

\newcommand{\bheading}[1]{{\noindent{\textbf{#1}}}}

\newcommand{\iheading}[1]{{\noindent{\textit{#1}}}}

\newcounter{note}[section]

\usepackage{algorithm}
\usepackage{algpseudocode}
\usepackage{xspace}

\newcommand{\secref}[1]{\mbox{Sec.~\ref{#1}}\xspace}
\newcommand{\ssecref}[1]{\mbox{\S\ref{#1}}\xspace}

\newcommand{\figref}[1]{\mbox{Fig.~\ref{#1}}}

\newcommand{\ignore}[1]{}

\newcommand{\ie}{\textit{i.e.}\xspace}
\newcommand{\eg}{\textit{e.g.}\xspace}

\newcommand{\TEE}{\textsc{TEE}\xspace}
\newcommand{\sysname}{\textsc{Raftel}\xspace}
\newcommand{\csysname}{\textsc{Chained-\sysname}\xspace}
\newcommand{\sysnameout}{\textsc{Engraft-\uppercase\expandafter{\romannumeral2}}\xspace}

\newcommand{\counter}{\textsc{Checker+}\xspace}

\newcommand{\damysus}{{Damysus}\xspace}

\newcommand{\hotstuff}{{HotStuff}\xspace}
\newcommand{\achilles}{{Achilles}\xspace}

\newcommand{\nv}{\textcolor{paramBlue}{\textsf{nv}}\xspace}
\newcommand{\prep}{\textcolor{paramBlue} {\textsf{prep}}\xspace}
\newcommand{\preco}{\textcolor{paramBlue}{\textsf{prec}}\xspace}
\newcommand{\com}{\textcolor{paramBlue}{\textsf{com}}\xspace}

\newcommand{\quorum}{{\textit{quorum}}\xspace}
\newcommand{\prepare}{{\textit{prepare}}\xspace}
\newcommand{\precommit}{{\textit{pre-commit}}\xspace}
\newcommand{\commit}{{\textit{commit}}\xspace}

\newcommand{\newview}{{\textit{new view}}\xspace}

\newcommand{\aif}{{\textbf{if}}\xspace}
\newcommand{\aelse}{{\textbf{else}}\xspace}
\newcommand{\athen}{{\textbf{then}}\xspace}
\newcommand{\aendif}{{\textbf{endif}}\xspace}
\newcommand{\return}{{\textbf{return}}\xspace}

\newcommand{\hash}{\textcolor{purple}{\textsf{H}}\xspace}
\newcommand{\leader}{\textcolor{purple}{\textsf{ldr}}\xspace}
\newcommand{\istee}{\textcolor{purple}{\textsf{isTEE}}\xspace}
\newcommand{\tmatch}{\textcolor{purple}{\textsf{T-match}}\xspace}
\newcommand{\mmatch}{\textcolor{purple}{\textsf{M-match}}\xspace}
\newcommand{\combine}{\textcolor{purple}{\textsf{Combine}}\xspace}
\newcommand{\teeprepare}{\textcolor{purple}{\textsf{TEEprepare}}\xspace}
\newcommand{\teeview}{\textcolor{purple}{\textsf{TEEview}}\xspace}
\newcommand{\teestore}{\textcolor{purple}{\textsf{TEEstore}}\xspace}
\newcommand{\teevote}{\textcolor{purple}{\textsf{TEEvote}}\xspace}

\newcommand{\priobroadcast}{\textcolor{purple}{\textsf{PrioBroadcast}}\xspace}
\newcommand{\safeproposal}{\textcolor{purple}{\textsf{SafeProposal}}\xspace}
\newcommand{\validqc}{\textcolor{purple}{\textsf{ValidQC}}\xspace}

\newcommand{\creatleaf}{\textcolor{purple}{\textsf{createLeaf}}\xspace}
\newcommand{\creatchain}{\textcolor{purple}{\textsf{createChain}}\xspace}
 
\newcommand{\msecs}{\ensuremath{\mathrm{ms}}\xspace}

\newcounter{packednmbr}

\newenvironment{packeditemize}{
\begin{list}{$\bullet$}{
\setlength{\labelwidth}{0pt}
\setlength{\leftmargin}{\labelsep}
\setlength{\itemsep}{1pt plus 0.5pt minus 0.5pt}
\setlength{\parsep}{0pt plus 0.5pt}
\setlength{\topsep}{1pt plus 0.5pt minus 0.5pt}
\setlength{\parindent}{0pt}
\setlength{\listparindent}{0pt}
}}{\end{list}}

\copyrightyear{2027}
\acmYear{2027}
\setcopyright{cc}
\setcctype{by}
\acmConference[EuroSys '27]{22nd European Conference on Computer Systems}{April 19--23, 2027}{Rabat, Morocco}
\acmBooktitle{22nd European Conference on Computer Systems (EuroSys '27), April 19--23, 2027, Rabat, Morocco}
\acmDOI{10.1145/3842654.3848557}
\acmISBN{979-8-4007-2971-3/2027/04}

\begin{document}

\title{Breaking Fault Lines: Unifying TEE-Assisted BFT Consensus in Partially Trusted Worlds}

\author{Xiaoqing Wen}
\email{xqwen@student.ubc.ca}
\affiliation{
  \institution{University of British Columbia}
  \country{Canada}
}

\author{Tong Liu}
\email{12332466@mail.sustech.edu.cn}
\affiliation{
  \institution{SUSTech}
  \country{China}
}

\author{Jianyu Niu}
\authornote{Part of the work was done when Xiaoqing Wen and Jianyu Niu visited
Research Institute of Trustworthy Autonomous Systems of SUSTech. Corresponding author: Jianyu Niu.}
\email{jianyniu@cityu.edu.hk}
\affiliation{
  \institution{City University of Hong Kong}
  \country{China}
}

\author{Jialin Li}
\email{lijl@comp.nus.edu.sg}
\affiliation{
  \institution{National University of Singapore}
  \country{Singapore}
}

\author{Cong Wang}
\email{congwang@cityu.edu.hk}
\affiliation{
  \institution{City University of Hong Kong}
  \country{China}
}

\author{Yinqian Zhang}
\email{yinqianz@acm.org}
\affiliation{
  \institution{SUSTech}
  \country{China}
}

\author{Chen Feng}
\email{chen.feng@ubc.ca}
\affiliation{
  \institution{University of British Columbia}
  \country{Canada}
}

\begin{abstract}

This paper revisits TEE-assisted BFT under a \textit{universal partial-TEE model}, where an arbitrary subset of replicas are equipped with TEEs while the remaining replicas operate without hardware trust guarantees.
We show that heterogeneous trust changes the structure of quorum formation and fault tolerance. In particular, we derive a tight resilience bound $f < \max \{\frac{n}{3}, \frac{m}{2} \}$, 
where $n$ is the total number of replicas and $m$ is the number of TEE-enabled replicas. The result reveals a sharp threshold phenomenon: TEEs improve fault tolerance only once they exceed two-thirds of the deployment.

Guided by this characterization, we introduce two protocol principles: (1) a dual-quorum construction that safely combines TEE-only and mixed quorums, and (2) a TEE-leader fast path that leverages hardware-enforced non-equivocation to reduce both consensus and view-change latency.
We realize these ideas in \sysname, which is, to our knowledge, the first HotStuff-style BFT protocol designed explicitly for arbitrary partial-TEE deployments, and in \csysname, a pipelined variant that further accelerates mixed-trust execution.
We implement both protocols atop Intel SGX and evaluate them in LAN and WAN environments. Our results show that \sysname achieves up to 625 TPS with sub-670\,ms latency in WAN settings, outperforming HotStuff by up to 308 TPS in throughput while approaching the performance of fully TEE-assisted protocols.

\end{abstract}

\ccsdesc[500]{Computer systems organization~Dependable and fault-tolerant systems and networks}

\keywords{BFT consensus, TEE, CFT consensus}

\maketitle
 
\section{Introduction}\label{sec:intro}
Byzantine fault-tolerant (BFT) consensus protocols lie at the heart of distributed systems such as blockchains~\cite{algorand, hanke2018dfinity}.
They allow a network of replicas to agree on a sequence of transactions despite adversarial behaviors by a subset of replicas~\cite{pbft, HotStuffYin2019}. 
This strong fault model provides robustness under adversarial behavior, but incurs substantial replication and communication overhead. 
To tolerate $f$ Byzantine faults, classical TEE-free BFT requires at least $3f + 1$ replicas—compared to $2f + 1$ for crash fault tolerance—and incurs heavy communication overhead, often involving three or more phases per decision~\cite{fab}.
These costs limit the scalability of BFT protocols in practice.

Trusted Execution Environments (TEEs), such as Intel SGX~\cite{sgx2013}, offer a promising way to bridge this gap.
By preventing equivocation through hardware guarantees, TEEs can reduce the quorum size (as low as $f + 1$) and communication rounds of BFT protocols.
In fact, prior TEE-assisted BFT protocols~\cite{yandamuri, damysus, achilles} have demonstrated near-CFT (Crash Fault Tolerant) performance while preserving Byzantine resilience.  

Unfortunately, most existing designs in the literature rely on a strong and increasingly unrealistic assumption: \textit{all} replicas are equipped with TEEs. 
In practice, modern deployments are inherently \textit{heterogeneous}. Cloud providers expose different trusted-computing technologies with varying capabilities and trust assumptions. TEE-enabled instances are often restricted to specific hardware types, regions, or pricing tiers. Large distributed systems evolve incrementally, making simultaneous migration of all replicas impractical. In permissionless and decentralized settings, requiring universal TEE provisioning further undermines openness and raises participation barriers.

Consequently, systems increasingly operate in a \textit{mixed-trust regime}, where only a subset of replicas can provide hardware-enforced guarantees. This setting fundamentally challenges existing consensus designs. Classical TEE-free BFT assumes that no replicas are trusted, while prior TEE-assisted protocols assume homogeneous trust across all replicas. Real deployments satisfy neither assumption. Instead, consensus protocols must operate under \textit{heterogeneous trust}, where TEE-enabled and non-TEE replicas coexist within the same protocol execution.
This raises a fundamental question: \textit{How should BFT consensus protocols exploit heterogeneous trust under arbitrary mixtures of TEE and non-TEE replicas?}

In this paper, we revisit TEE-assisted BFT under a \textit{universal partial-TEE model} that generalizes both classical TEE-free BFT~\cite{HotStuffYin2019, pbft} and fully TEE-assisted consensus~\cite{Kapitza:2012:CheapBFT, achilles, damysus}. We consider a deployment of $n$ replicas, among which only $m$ replicas are TEE-enabled. Within this model, we derive a tight fault-tolerance bound: $$f < \max\left\{\frac{n}{3}, \frac{m}{2}\right\},$$ where $f$ denotes the number of Byzantine replicas.

This characterization reveals an important and previously overlooked phenomenon: partial TEE deployment does not improve resilience strictly monotonically. Instead, the system exhibits a sharp threshold effect.
When TEE-enabled replicas constitute at most two-thirds of the deployment (\ie, ${m}/{2} \le {n}/{3}$), fault tolerance is still fundamentally bounded by classical TEE-free BFT limits; TEEs can improve efficiency, but not resilience.
Only once TEE-enabled replicas exceed this threshold (\ie ${m}/{2} > {n}/{3}$) does hardware-enforced non-equivocation directly increase fault tolerance.

Guided by this characterization, we identify two protocol principles for consensus under heterogeneous trust:

\begin{packeditemize}

\item \textbf{Dual-quorum construction.} Because TEE-enabled replicas cannot equivocate, safety can be established using either a TEE-only quorum or a classical mixed quorum. We therefore introduce a dual-quorum design that safely combines:
(i) a TEE-Quorum consisting exclusively of TEE votes, and
(ii) a Mixed-Quorum consisting of arbitrary replicas.
When sufficient TEE replicas are available, the protocol commits using smaller and faster TEE-Quorums; otherwise, it safely falls back to classical TEE-free BFT quorum formation.

\item \textbf{TEE-leader fast path.} 
Prior TEE-assisted and hybrid BFT protocols have shown that a non-equivocating trusted leader can reduce the number of consensus phases by eliminating defenses against leader equivocation. We extend this optimization to arbitrary partial-TEE deployments, where TEE and non-TEE replicas---and consequently TEE and non-TEE leaders---may coexist within the same execution.
In \sysname, TEE-enabled leaders follow a shortened commit path enabled by hardware-enforced non-equivocation, while non-TEE leaders use the standard multi-phase BFT path. Moreover, \sysname leverages TEE-certified state to accelerate view changes following TEE-led views, avoiding redundant reconstruction of the safe extension point when such a certificate is available. Both execution paths share the same quorum and safety framework, allowing the protocol to transition safely between them as leadership changes across views.

\end{packeditemize}

We realize these ideas in \sysname\footnote{\sysname is a mythical island in One Piece~\cite{OnePiece}, where all voyages converge and the ultimate truth uniting the world is said to reside.}, to the best of our knowledge, the first HotStuff-style BFT protocol designed explicitly for arbitrary partial-TEE deployments.
\sysname builds on the chained structure and rapid leader rotation of HotStuff~\cite{HotStuffYin2019} and \damysus~\cite{damysus}, while introducing heterogeneous quorum formation and TEE-aware fast paths.
To support these mechanisms efficiently, we consolidate Damysus’s trusted components into a unified trusted module, \counter, which minimizes enclave transitions while supporting TEE-specific protocol logic.

We further propose \csysname, a pipelined variant that extends the chained execution model. In particular, \csysname introduces a \textit{pipelined commit rule} in which blocks proposed by non-TEE leaders can be committed earlier once extended by TEE-led proposals. This optimization reduces latency even in mixed-trust deployments where only a subset of leaders are TEE-enabled.

We implement \sysname and \csysname atop \damysus using Intel SGX and evaluate them in both LAN and WAN environments. We compare against \hotstuff, \damysus, and \achilles across deployments with up to $f=32$ Byzantine faults. Our experimental results show that \sysname achieves up to 625 TPS with sub-670\,ms latency in WAN settings. Compared with \hotstuff, \sysname improves throughput by up to $1.02\times$ while reducing latency by 61\%, approaching the performance of fully TEE-assisted protocols.

\bheading{Contributions.} The main contributions are as follows:

\begin{packeditemize}
\item Universal partial-TEE model. We formalize a universal model for TEE-assisted BFT consensus under arbitrary mixtures of TEE and non-TEE replicas, and derive a tight fault-tolerance bound characterizing consensus under heterogeneous trust.

\item Design principles for heterogeneous trust. We introduce a dual-quorum construction and TEE-aware fast paths that exploit hardware-enforced non-equivocation while remaining safe under arbitrary partial deployment.

\item Protocol design and implementation. We design and implement \sysname and \csysname, the first HotStuff-style protocols supporting arbitrary partial-TEE deployments.

\item Evaluation. We evaluate our prototype on Intel SGX in LAN and WAN settings, demonstrating substantial performance improvements over classical TEE-free BFT protocols while approaching the performance of fully TEE-assisted systems.

\end{packeditemize}

\section{Related Work and Motivation}\label{sec:related}

\subsection{Related Work}
We survey prior work on BFT consensus, focusing on TEE-assisted BFT protocols and modern hybrid BFT protocols.

\bheading{TEE-assisted BFT consensus.} 
TEE-assisted protocols exploit hardware-enforced non-equivocation to reduce quorum sizes and communication phases. 
Early work, such as Hybster~\cite{Behl:2017:HSS}, uses trusted counters within TEEs to parallelize consensus instances, while FastBFT~\cite{fastBFT} leverages TEEs to implement secret sharing and achieve $O(n)$ communication complexity. 
\damysus~\cite{damysus} builds atop modern BFT designs such as HotStuff~\cite{HotStuffYin2019}, enabling three-phase commit with linear communication complexity.
Subsequent works, including OneShot~\cite{decouchant2024oneshot}, FlexiBFT~\cite{gupta2022dissecting}, and \achilles~\cite{achilles}, further reduce the number of communication phases or improve protocol efficiency.
Another line of work executes the entire transaction processing pipeline inside TEEs to provide confidentiality guarantees~\cite{russinovich:2019:ccf, ENGRAFT, HyperSrds}. 
Our work focuses on protocols that leverage the integrity guarantees of TEEs, although the theoretical insights derived here can also apply to confidentiality-oriented systems.

{\bheading{Hybrid fault models.}
Prior hybrid-fault models classify replicas according to failure semantics (\eg, crash versus Byzantine behavior). Early work~\cite{meyer2002consensus, thambidurai1988interactive} introduced dual-failure models and proposed protocols tolerating bounded combinations of crash and Byzantine faults. UpRight~\cite{upright} further showed that distinguishing crash faults from Byzantine faults can improve fault tolerance compared to traditional BFT models. Scrooge~\cite{scrooge} explored mixed crash/Byzantine settings to reduce the cost of fast Byzantine replication in the presence of unresponsive replicas through replier quorums and message histories. 
Subsequent systems investigated additional hybrid assumptions. VFT~\cite{vft}, XFT~\cite{xft}, and RR~\cite{rr} study orthogonal dimensions, including correlated failures, network synchrony, partition tolerance, and rollback attacks. However, these systems assume homogeneous replicas under a uniform fault model, where all replicas are treated identically and may potentially exhibit the same classes of faults. As a result, protocol design cannot exploit replica-specific identities or hardware-constrained behaviors to optimize consensus execution.

SeeMoRe~\cite{seemore} adopts a mixed-trust deployment model in which private-cloud replicas are trusted (crash-only), while public-cloud replicas may behave Byzantine, and leverages trusted leaders to reduce communication phases. However, its optimization mainly relies on trusted leadership and does not fundamentally change quorum construction.}

\bheading{Hybrid TEE systems.} 
Recent works have recognized the practicality of partial and incremental adoption~\cite{gao2022mixed, king2023parteetor, sinha2019luciditee}. 
ParTEETor~\cite{king2023parteetor} demonstrates that even limited TEE penetration in Tor can enhance resistance against deanonymization without degrading performance.
In secure multiparty computation, LucidiTEE~\cite{sinha2019luciditee} shows that fair exchange can be achieved with only $t$ out of $n$ TEE-enabled participants, rather than requiring TEEs for all participants as in~\cite{fairMPC}, while tolerating up to $t$ ($t < n$) malicious parties.
Inspired by this line of work, we study BFT consensus under partial TEE deployment; whereas prior efforts primarily target security, our work also emphasizes efficiency.
Closer to our setting, the Mixed Fault Tolerance (MFT) protocol~\cite{gao2022mixed} studies BFT consensus with partially TEE-enabled replicas, in which $n = 3f+2$ replicas are required to tolerate $f$ faults for leader election safety---a stricter bound than the conventional $3f+1$ bound. 
However, MFT does not provide a systematic analysis of how the number of TEE-enabled replicas affects the fault-tolerance bound. 
Besides, MFT does not fully exploit the availability of TEEs to optimize protocol efficiency; in particular, it overlooks opportunities such as dual-quorum formation and fast commit rules. 

\subsection{Why a Universal Partial-TEE Model?} \label{subsec:motivation}

Existing TEE-assisted BFT protocols assume that every replica is equipped with a TEE. 
This assumption is often unrealistic in practice, where deployments are heterogeneous and only a subset of replicas may have access to TEE support. 

\begin{packeditemize}

{
    \item \textit{TEE-enabled devices are not yet universally deployable.} 
    TEEs are increasingly available across both cloud and edge platforms, including server-grade TEEs such as Intel SGX/TDX and AMD SEV-SNP, as well as Arm TrustZone on smartphones and mobile devices\footnote{Intel SGX has been deprecated on client-class processors~\cite{sgx-intel}.}. However, TEE support is still limited to specific processor generations, server models, cloud instance types, or deployment configurations. 
    Major cloud providers, including AWS, Alibaba Cloud, Microsoft Azure, and Google Cloud, provide limited TEE-enabled offerings. For example, AWS currently offers AMD SEV-SNP only on a few instance families (\texttt{M6a}, \texttt{C6a}, \texttt{R6a}) in select regions, often with additional cost overhead~\cite{AWS}.

    \item \textit{Incremental deployment is necessary.} Real-world systems evolve gradually. Migrating from non-TEE to fully TEE-enabled infrastructures takes time, requiring both hardware replacement and software redesign. This process naturally creates transitional stages where only a subset of replicas are TEE-enabled, making partial deployments both common and realistic~\cite{king2023parteetor}.

    \item \textit{Hybrid infrastructure naturally arises in practical distributed deployments.} 
    Many consortium, partially decentralized, and large-scale distributed systems operate across organizations and infrastructures with different hardware capabilities and operational constraints~\cite{king2023parteetor,buford2009p2p}. In practice, some replicas may provision TEEs while others rely on commodity hardware due to differences in cost, deployment policies, hardware availability, or operational requirements. Consequently, these systems naturally form mixed-replica environments rather than uniformly TEE-enabled deployments.

  }

\end{packeditemize}

These observations motivate a \textit{universal} TEE-assisted BFT protocol that generalizes both classical TEE-free BFT and fully TEE-assisted BFT, while capturing the mixed deployments common in practice.

\section{System Model and Goals}

\subsection{System Model}\label{subsec:model}
We consider a system of $n$ replicas, among which $m$ replicas are equipped with TEEs and the remaining $n-m$ replicas execute without TEEs. 
We denote the set of TEE-enabled replicas by $S_{\TEE}$, where $|S_{\TEE}|=m$.
We refer to replicas in $S_{\TEE}$ as \textit{TEE replicas} and all others as \textit{non-TEE replicas}.
TEE replicas provide hardware-enforced execution integrity and non-equivocation guarantees, whereas non-TEE replicas may behave arbitrarily.
We assume that each replica's TEE status remains fixed during a protocol execution.

The model captures both classical TEE-free BFT and fully TEE-assisted BFT as special cases. 
When $m=0$, no replica is TEE-enabled, and the setting reduces to classical TEE-free BFT. 
When $m=n$, all replicas are TEE-enabled, and the setting reduces to fully TEE-assisted BFT. Our focus is the general case $0<m<n$, where replicas operate under heterogeneous trust.

\bheading{Cryptographic setup.}
We assume a standard public-key infrastructure (PKI) for authenticating replicas and messages.
Each replica $p_i$ has a public/private key pair $(pk_i,sk_i)$.
For TEE replicas, the signing (private) key used by the trusted component is generated and stored inside TEEs and cannot be extracted by the untrusted host.
TEE replicas also support remote attestation~\cite{johnson2016intel, MAGE},
allowing other replicas to verify both the identity of a TEE replica and the code executed inside its trusted component. Thus, all replicas can reliably determine whether a message was produced by a TEE-backed trusted component or by ordinary replica software.

\bheading{Network model.}
We adopt the standard partially synchronous network model~\cite{dwork1988consensus} used by many BFT protocols~\cite{pbft, HotStuffYin2019, fastHotStuff}. 
Before the unknown Global Stabilization Time (\textsf{GST}), messages may be delayed arbitrarily.
After GST, there exists a known bound $\Delta$ such that every message sent between honest replicas is delivered within $\Delta$.
Replicas communicate over authenticated channels, and the adversary may delay, drop, reorder, or inject messages subject to cryptographic unforgeability.

\bheading{Fault model.}
The system tolerates up to $f$ Byzantine replicas, chosen from both TEE and non-TEE ones.
Byzantine replicas may collude. 
For ease of presentation, we consider a Byzantine adversary that controls all Byzantine replicas.
We distinguish three cases.

\begin{packeditemize}
    \item \textbf{Honest replicas.} Honest replicas, whether TEE-enabled or not, follow the protocol.

    \item \textbf{Byzantine non-TEE replicas.}  A corrupted non-TEE replica may equivocate, forge local state, send conflicting protocol messages, omit messages, or otherwise behave arbitrarily, subject only to standard cryptographic assumptions~\cite{pbft, HotStuffYin2019}.

    \item \textbf{Byzantine TEE replicas.} A corrupted TEE replica has an adversarial host operating system and untrusted application environment.
    The adversary may schedule, delay, replay, or reorder inputs and outputs to trusted components, and may invoke them with arbitrary inputs.
    However, it cannot extract secrets from the TEE, forge TEE-generated signatures, or cause the trusted component to execute code other than the attested protocol logic~\cite{damysus, achilles, gai2021dissecting, decouchant2024oneshot}. 
 
\end{packeditemize}

{
\bheading{TEE assumptions and non-goals.}
Our protocol relies on the integrity of the trusted component and the confidentiality of keys stored inside it.
We do not attempt to defend against transient-execution attacks~\cite{Schwarz:2019:zombieload, van:2019:ridl, chen:2019:sgxpectre}, micro-architectural side-channel attacks~\cite{xu2015controlled, shinde2016preventing, wang2017leaky, van2017telling, werner2019severest}, or rollback/forking attacks~\cite{Rote, narrator}, as these are largely orthogonal to our protocol design. 
For example, rollback and forking attacks can be addressed using state-continuity mechanisms~\cite{Rote, 2023ccf, 2025recipe}, which ensure monotonic evolution of trusted state.
Such mechanisms may add local storage or cryptographic overhead, but do not change the quorum structure or communication pattern analyzed in this paper.
}

\subsection{System Goals}

Clients submit transactions to replicas, which order them into a sequence of blocks. Each block extends a parent block, forming a hash-linked chain. The protocol must satisfy:

\begin{packeditemize}
    \item \textsl{Safety:} If two honest replicas commit two blocks $b$ and $b^{\prime}$ at the same height, then $b = b^{\prime}$. 

    \item \textsl{Liveness:} After \textsf{GST}, every transaction submitted by an honest client is eventually included in a block committed by honest replicas.
    
\end{packeditemize} 

\section{Pushing the Limits: Fault Tolerance Bound and Design Principles} \label{sec:bounds}

We now characterize the resilience limits of consensus under heterogeneous trust and derive the protocol principles that guide \sysname. The key question is how the availability of $m$ TEE replicas (among $n$ total replicas) affects both fault tolerance and efficient quorum formation.

\subsection{Fault-Tolerance Bound}\label{subsec:bound}

We first recall the two homogeneous extremes. In classical TEE-free BFT, where no replica is trusted, consensus requires $n \ge 3f+1$ replicas to tolerate $f$ Byzantine faults~\cite{lamport1982byzantine}. That is, $f < n/3$ when $m=0$. In fully TEE-assisted BFT, where all replicas are equipped with trusted components that prevent equivocation, consensus can tolerate up to a minority of Byzantine replicas, requiring only $n \ge 2f+1$ replicas~\cite{clement2012limited}. That is, $f < n/2$ when $m = n$.

The partial-TEE setting lies between these extremes but is not obtained by simply interpolating between them. We will show that consensus under heterogeneous trust has the following \textbf{tight resilience bound}:
\[
f < \max \!\left\{ \tfrac{n}{3}, \tfrac{m}{2}\right\}.
\]

This bound says that, to tolerate $f$ Byzantine faults, it suffices that either
\[
n \ge 3f+1 \quad \text{or} \quad m \ge 2f+1.
\]
The first condition corresponds to classical TEE-free BFT: even without enough TEE replicas, safety can be maintained using mixed quorums over the full replica set. The second condition corresponds to a \textit{TEE-dominated regime}: if sufficiently many replicas are TEE-enabled, their non-equivocation guarantees alone can support safe quorum formation.

This sharp transition is central to our design. It implies that TEEs should be used in two distinct ways: to form smaller quorums when enough TEE replicas are available, and to accelerate protocol phases whenever a TEE leader can prevent equivocation.

\bheading{Proof intuition.}
To see why this resilience bound is necessary, we extend the classical
indistinguishability argument for Byzantine agreement~\cite{lamport1982byzantine}
to the heterogeneous-trust setting.
We partition the replicas into three non-empty sets $S_1$, $S_2$, and $S_3$,
each of size at most $f$, where $S_3$ contains only non-TEE replicas.
Before GST, the adversary delays communication between $S_1$ and $S_2$ and
corrupts $S_3$.
Because $S_3$ can equivocate, it can make $S_1$ and $S_2$ observe two
indistinguishable but conflicting executions, forcing them to decide
different values and violating agreement.
The full proof appears in Appendix~B of the
\href{https://github.com/1wenwen1/Raftel/blob/main/doc/Breaking_Fault_Lines__Unifying_BFT_Consensus_in_a_Partially_Trusted_World.pdf}{extended version}. The bound is also achievable, as demonstrated by the protocol presented later in this paper.

\subsection{Principle 1: Dual-Quorum Construction}\label{subsec:dualqc}

The resilience bound suggests that a protocol for heterogeneous trust should not rely on a single quorum rule. Instead, it should support two quorum types.

\begin{packeditemize}
    \item \textbf{TEE-Quorum}: A TEE-Quorum consists of at least $Q_T$ votes, all from TEE replicas, where $Q_T = \max \{ \left\lfloor \tfrac{m}{2} \right\rfloor+1, f+1\}$. 
    The majority term ensures that any two TEE-Quorums intersect in at least one TEE replica. The $f+1$ term ensures that the quorum contains at least one honest TEE replica. Since TEE-backed votes are non-equivocating even when the host is Byzantine, the intersection is sufficient to prevent conflicting certificates.
    \item \textbf{Mixed-Quorum}: A Mixed-Quorum consists of at least $Q_M$ votes from arbitrary replicas,
    where $Q_M = n-f $. This is the standard BFT quorum size. Any two Mixed-Quorums intersect in at least $f+1$ replicas. The intersection contains at least one honest replica when $n \ge 3f+1$.
\end{packeditemize}

\bheading{Safety across quorum types.}
The key requirement is not merely that each quorum type is safe in isolation, but that different quorum types remain safe when used interchangeably. Under the bound above, any two valid quorums---TEE/TEE, mixed/mixed, or TEE/mixed---intersect in at least one replica that cannot equivocate. This property allows the protocol to form certificates using whichever quorum becomes available first, while preserving a single global safety argument.

This dual-quorum construction lets the protocol adapt to heterogeneous deployments. When enough TEE replicas respond quickly, the leader can form a smaller TEE certificate. When TEE replicas are slow, unavailable, or insufficient, the protocol falls back to a classical Mixed-Quorum without changing the safety rule.
The formal quorum construction and correctness proof are presented in~\secref{sec:protocol}.

\subsection{Principle 2: TEE-Leader Fast Path}\label{subsec:tee-leader}

The second design principle exploits a different consequence of trusted hardware: a TEE-enabled leader cannot equivocate.

\iheading{1) TEE-leader fast path.}
Classical TEE-free BFT protocols require additional phases to protect against a Byzantine leader that proposes conflicting blocks in the same view. For example, prepare-style phases ensure that replicas observe a consistent proposal before committing.
By contrast, a TEE-enabled leader can produce at most one valid proposal for a given view, and the proposal is bound to the attested protocol logic. Thus, the protocol can safely bypass phases whose sole purpose is to defend against leader equivocation.
In \sysname, this yields a fast path for TEE leaders: consensus can complete with fewer communication phases, while non-TEE leaders use the standard multi-phase path. Thus, the protocol remains safe under arbitrary leader schedules but obtains lower latency whenever the current leader is TEE-enabled.

\iheading{2) Fast view change.}
TEE-based non-equivocation can accelerate view change when the trusted component additionally certifies sufficient safety information about the previous view.  
In classical TEE-free BFT, a new leader typically collects new-view messages or the highest prepared certificate to determine which block is safe to extend, because the previous leader may have equivocated. 
If the previous leader was TEE-enabled, however, there can be at most one valid proposal from that view.

Therefore, the next leader can safely inherit the unique TEE-certified proposal rather than reconstructing safety solely from a full set of new-view messages. This reduces leader-change latency and improves responsiveness under frequent leader rotation. Importantly, this optimization is conditional: when the previous leader is not TEE-enabled, the protocol falls back to the standard view-change rule.

\subsection{Implications}

Together, these principles translate the fault-tolerance characterization into protocol structure. The dual-quorum construction exploits heterogeneous trust at the quorum level, while the TEE-leader fast path and fast view change exploit non-equivocation at the leader level. This separation is important: even when TEEs are too sparse to improve resilience, they can still improve performance by reducing quorum size, shortening the commit path, or accelerating leader transitions.
The next sections instantiate these principles in \sysname and its pipelined variant, \csysname.

\section{\sysname Design}\label{sec:protocol}

\begin{figure}[t]
\centering
\includegraphics[width=8.5cm]{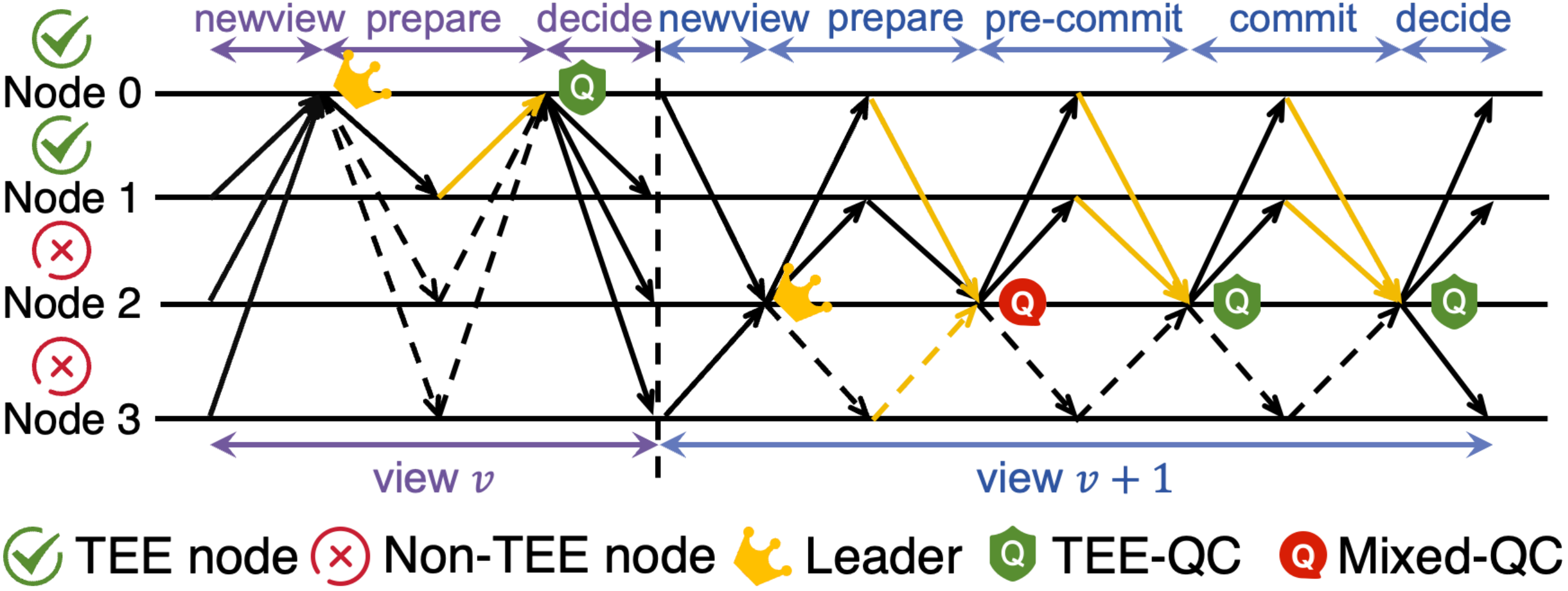}
\caption{\sysname Overview. Solid arrows denote prioritized broadcasts to TEE replicas, while dashed arrows denote subsequent broadcasts to non-TEE replicas. Yellow arrows highlight the messages that arrive at the leader first.}
\Description{Message-sequence diagram of Raftel across two consecutive views with four replicas. Nodes 0 and 1 are TEE-enabled; nodes 2 and 3 are not. In the first view, TEE leader node 0 proceeds through new-view, prepare, and decide. In the next view, non-TEE leader node 2 additionally uses pre-commit and commit. Solid and dashed arrows distinguish prioritized TEE broadcasts from subsequent non-TEE broadcasts; highlighted arrows indicate the earliest arrivals at the leader. Quorum-certificate symbols distinguish TEE-only and mixed quorums.}
\label{fig:overview}

\end{figure}

\subsection{Overview}\label{sec:overview}
\sysname is a mixed-trust TEE-assisted BFT protocol that operates under arbitrary mixtures of TEE-enabled and non-TEE replicas. Rather than relying on a binary fallback between a TEE-based protocol and a non-TEE protocol, \sysname is designed as a unified protocol for mixed-replica settings that can continuously adapt to different levels of TEE availability.

At a high level, all replicas follow the same protocol structure, but differ in how they handle security-critical steps. 
TEE replicas invoke a trusted component called \counter\ on critical paths to generate certified messages and enforce non-equivocation constraints across protocol steps. 
Replicas without TEEs execute these steps entirely outside trusted hardware, without the additional guarantees provided by \counter.
As illustrated in \figref{fig:overview}, \sysname combines two key ideas introduced in the previous section: (1) a dual-quorum, and (2) a TEE-leader fast path.

The dual-quorum design (\ssecref{subsec:dualqc}) enables \sysname to dynamically switch between two quorum formation modes. 
When sufficient TEE-backed votes are available, the protocol forms a TEE-Quorum using only TEE replicas, allowing smaller quorum certificates and faster progress. 
Otherwise, the protocol safely falls back to a Mixed-Quorum formed from arbitrary replicas, preserving compatibility with classical TEE-free BFT quorum formation under mixed-trust deployments.

Hybrid trust also enables leader-dependent commit paths (\ssecref{subsec:tee-leader}). 
When the leader is TEE-enabled, the trusted component constrains certified proposals generated by the leader, substantially reducing the coordination overhead required to tolerate equivocation. 
This allows the protocol to shorten the critical commit path and accelerate leader transitions. 
When the leader is not TEE-enabled, \sysname safely falls back to the standard BFT commit path.

\bheading{Pacemaker.}
\sysname adopts a standard partially synchronous pacemaker similar to those used in HotStuff-style BFT protocols~\cite{HotStuffYin2019,diembft2021}. 
The pacemaker is responsible for ensuring liveness after GST, while safety is guaranteed independently through quorum intersection and TEE-enforced non-equivocation. 
Our protocol design is orthogonal to the pacemaker and does not modify its underlying mechanism.

\subsection{Definitions}
\bheading{Views, phases, and steps.}
\sysname proceeds in views, each coordinated by a unique leader known to all replicas. 
We denote the leader of view $v$ by $\leader(v)$. 
Replicas may be either TEE-enabled or non-TEE replicas. 
We use $\istee(id)$ to denote whether replica $id$ is TEE-enabled, which can be verified through remote attestation~\cite{johnson2016intel}. 
Each view consists of multiple phases identified by a phase tag 
$ph \in \{\nv, \prep, \preco, \com\}$, corresponding to 
\textsf{new-view}, \textsf{prepare}, \textsf{pre-commit}, and \textsf{commit}, respectively. 
We define a \textit{step} as a pair $(v, ph)$ representing phase $ph$ in view $v$.

For a non-TEE leader, a view progresses sequentially through all phases:
$(v, \nv) \rightarrow (v, \prep) \rightarrow (v, \preco) \rightarrow (v, \com) \rightarrow (v+1, \nv)$.
If the leader of view $v$ is a TEE replica, the protocol skips the $\preco$ and $\com$ phases, and the progression becomes:
$(v, \nv) \rightarrow (v, \prep) \rightarrow (v+1, \nv)$. 
We use $(v, ph)++$ to denote the transition to the next phase.

\bheading{Blocks and chains.} 
A block contains client transactions and the hash of its parent block. 
Blocks are linked through hash references to form a chain rooted at the genesis block $\mathcal{G}$. 
We write $b_1 \succ b_2$ if block $b_1$ extends block $b_2$. For a block $b$, $b.\mathit{parent}$ denotes the hash of its parent block. Similarly, $b_1 \succ h$ denotes that $b_1$ extends the block with hash value $h$.
Two blocks conflict if neither extends the other. 
The height of a block is its distance from the genesis block. 
We compare the freshness of blocks using their associated views, and blocks proposed in higher views are considered more recent. 
We assume a function \creatleaf($h, txs$) that creates a new block extending the parent block with hash $h$.

\bheading{Certificates.}
Replicas authenticate protocol messages using digital signatures. 
We use $\langle msg \rangle_{\sigma}$ to denote a signed message carrying signature $\sigma$, and $\langle msg \rangle_{\vec{\sigma}}$ to denote a message carrying a set of signatures $\vec{\sigma}$. 

A \textit{certificate} is a signed statement of the form
$\phi = \langle h, v, h^\prime, v^\prime,\\ ph \rangle_{\vec{\sigma}}$,
where $h$ and $v$ denote the hash and view of the certified block, $h^\prime$ and $v^\prime$ denote the hash and view of the justified block (\ie, a block can be safely extended), and $ph$ denotes the phase. 
Given a certificate $\phi$ of the form $\langle h, v, h^{\prime}, v^{\prime}, ph\rangle_{\vec{\sigma}}$, let $\phi._{Hprep}$ be $h$; $\phi._{Vprep} $  be $v$; $\phi._{Hjust}$ be $h^{\prime}$; $\phi._{Vjust} $  be $v^{\prime}$; and $\phi._{sign}$ be $\vec{\sigma}$. We use $\vec{\phi}$ to denote a list of certificates, and $\vec{\phi}^n$ to indicate that the list has length $n$. 
We use $\perp$ for unused fields when a value is not required.
Certificates may be generated either by the trusted component of a TEE-enabled replica or by the local software logic of a non-TEE replica.

\bheading{Quorum certificates.}
A \textit{Quorum Certificate} (QC) is an aggregated proof that a quorum of replicas voted for the same block in the same view and phase. 
Following \damysus~\cite{damysus}, QCs are constructed by aggregating partial certificates using multi-signatures:
$\combine\big([\langle h, v, h^\prime, v^\prime, ph \rangle_{\sigma_1}, \ldots, \langle h, v, h^\prime, v^\prime,\\ ph \rangle_{\sigma_n}]\big) = \langle h, v, h^\prime, v^\prime, ph \rangle_{[\sigma_1, \ldots, \sigma_n]}$.

\sysname supports two types of quorum certificates corresponding to the dual-quorum design:

\begin{packeditemize}
    \item \textbf{TEE-QC:} $qc \leftarrow \combine(\vec{\phi}^{Q_T})$, a list of $Q_T$ certificates issued exclusively by TEE replicas. 
    
    \item \textbf{Mixed-QC:} $qc \leftarrow  \combine(\vec{\phi}^{Q_M})$, a list of $Q_M$ certificates from any combination of TEE and non-TEE replicas. 
\end{packeditemize}

Given a list of certificates $\vec{\phi} = [\phi_1, ..., \phi_n]$, we define the following operations to check whether \quorum certificates have been received:
let \tmatch($\vec{\phi}, k, h, v, ph$) be true iff: (1) $n=k$; (2) all $n$ signatures have been created by different TEE replicas; and (3) $\forall i \in \{1,...,  n\}, h = h_i \land v =v_i \land ph = {ph}_i$.
Similarly, $\mmatch(\vec{\phi}, k, h, v, ph)$ is defined with the same conditions except that in (2) the $n$ signatures may come from different \textit{arbitrary replicas} (TEE or non-TEE).

The predicate $\validqc(\cdot)$ accepts either a certificate vector $\vec{\phi}$ or a multi-signature certificate $\phi$. It verifies that all signatures are valid and come from distinct replicas, with at least $Q_T$ TEE signers or at least $Q_M$ signers in total. For a certificate vector, each signature is verified against its corresponding message; for a multi-signature certificate, all signatures must cover the same message.
The operation $\priobroadcast(msg)$ sends $msg$ first to TEE replicas and then to the remaining replicas, favoring TEE-Quorum formation when possible.
Finally, $\safeproposal(b,\phi_{prep})$ validates the proposal and its justification, checks that $b$ extends the justified parent, and requires that $b$ extend the locally locked block or that the justified prepared view exceed the locked view.

\begin{algorithm*}[t]

\caption{The pseudocode of operations for replica $i$ in \sysname}
\label{alg:normal-op}

\noindent
\begin{minipage}[t]{0.49\linewidth}
\begin{algorithmic}[1] 
\State $pks$ \Comment{public keys}
\State $(view, phase) = (0, \nv)$  \Comment{current view and phase}
\State $(lockv, lockh) := (0, H(\mathcal{G}))$\Comment{latest locked block}
\State $qc_{prep} := \langle H(\mathcal{G}),0,\perp,0,\prep\rangle$ \Comment{last prepared certificate}
\State
\State  \Comment{prepare phase}
\State  \textbf{as a leader} 
\State  \hspace{1em} wait for $\vec{\phi}$ s.t. \tmatch $(\vec{\phi}, Q_T, \perp, v, h^{\prime}, v^{\prime}, \nv) $
\Statex \hspace{6em} $ \lor ~ \mmatch (\vec{\phi}, Q_M, \perp, v, h^{\prime}, v^{\prime}, \nv)$
\State  \hspace{1em} $\phi_{nv} :=$ certificate $\phi \in \vec{\phi}$ with highest $\phi._{Vjust}$ 
\State  \hspace{1em} $b :=$ \creatleaf$(\phi_{nv}._{Hjust}, txs)$
\State  \hspace{1em} \aif $\istee(i)$ \athen
\State  \hspace{1em} \hspace{1em} $\phi := \teeprepare(b, H(b),\phi_{nv}, \vec{\phi})$
\State  \hspace{1em} \hspace{1em} \priobroadcast $\langle \phi_{nv}, b, \phi \rangle$
\State  \hspace{1em} \aelse
\State  \hspace{1em} \hspace{1em} \priobroadcast ($\langle \phi_{nv}, b, \langle  \hash(b),  v, \prep\rangle_{\sigma} \rangle$)
\State  \hspace{1em} \aendif
\State  \textbf{all replicas} 
\State  \hspace{1em} wait for $\langle \langle  h^\prime, v^\prime, \nv \rangle_{\sigma}, v, b, \prep \rangle_{\sigma^\prime}$ from the leader

\State  \hspace{1em} $\phi_{nv} := \langle h^\prime, v^\prime, \nv \rangle_{\sigma}$
\State  \hspace{1em} $\phi_{prep} := \langle H(b), v, h^\prime, v^\prime, \prep \rangle_{\sigma^\prime}$
\State  \hspace{1em} \aif $\istee(i)$ \athen
\State  \hspace{2em} $\phi^\prime := \teevote(b, \phi_{prep})$
\State  \hspace{1em} \aelse
\State  \hspace{2em} \textbf{abort if} $\neg(\safeproposal (b,\phi_{prep}) $
\State  \hspace{6em} $\land view = v \land phase = \prep )$
\State  \hspace{2em} \aif $\istee(\leader(view))$ \athen 
\State  \hspace{3em} $\phi^\prime := \langle \hash(b), view, h^\prime, v^\prime, \com \rangle_{\sigma_i}$ 
\State  \hspace{2em} \aelse $\phi^\prime := \langle  \hash(b), view, h^\prime, v^\prime, \prep \rangle_{\sigma_i}$
\State  \hspace{2em} \aendif
\State  \hspace{2em}  $qc_{prep} = \phi_{prep}$
\State  \hspace{1em} \aendif
\State  \hspace{1em}  send $\phi^\prime$ to leader; $(view, phase)++$  

\State
\State  \Comment{pre-commit phase}
\State  \textbf{as a leader} 
\State  \hspace{1em} wait for $\vec{\phi}$ s.t. \tmatch $(\vec{\phi}, Q_T, h, v, h^{\prime}, v^{\prime}, \prep) $
\Statex \hspace{6em} $ \lor ~ \mmatch (\vec{\phi}, Q_M, h, v, h^{\prime}, v^{\prime}, \prep)$
\State  \hspace{1em} \priobroadcast($qc_{prep} :=$ \combine($\vec{\phi}$))

\State
\State  \textbf{all replicas}
\State  \hspace{1em} wait for $\langle h, v, h^{\prime}, v^{\prime}, \prep \rangle_{\vec{\sigma}}$ from the leader
\State  \hspace{1em} $\phi_{prep} = \langle h, v, h^{\prime}, v^{\prime}, \prep \rangle_{\vec{\sigma}}$

\algstore{mysplit}              
\end{algorithmic}
\end{minipage}\hfill
\begin{minipage}[t]{0.49\linewidth}
\begin{algorithmic}[1]          
\algrestore{mysplit}

\State  \hspace{1em} \aif $\istee(i)$ \athen $\phi^\prime := \teestore(\phi_{prep})$
\State  \hspace{1em}  \aelse 
\State  \hspace{2em} \textbf{abort if} $\neg(\validqc (\phi_{prep})\land view = v$ 
\Statex \hspace{7em} $\land phase = \preco \land h=\hash(b))$
\State  \hspace{2em}  $\phi^\prime := \langle  \hash(b), view, \preco \rangle_{\sigma_i}$
\State  \hspace{1em} \aendif
\State  \hspace{1em}  $qc_{prep} = \phi_{prep}$
\State  \hspace{1em} send $\phi^\prime$ to leader; $(view, phase)++$

\State
\State  \Comment{commit phase}
\State  \textbf{as a leader} 
\State  \hspace{1em} wait for $\vec{\phi}$ s.t. \tmatch $(\vec{\phi}, Q_T, h, v, h^{\prime}, v^{\prime}, \preco) $
\Statex \hspace{6em} $ \lor ~ \mmatch (\vec{\phi}, Q_M, h, v, h^{\prime}, v^{\prime}, \preco)$
\State  \hspace{1em} \priobroadcast ($qc_{prec} :=$ \combine($\vec{\phi}$))

\State  \textbf{all replicas}
\State  \hspace{1em} wait for $\langle h, v, h^{\prime}, v^{\prime}, \preco \rangle_{ \vec{\sigma}}$ from the leader
\State  \hspace{1em} $\phi_{prec} = \langle h, v, h^{\prime}, v^{\prime}, \preco \rangle_{ \vec{\sigma}}$
\State  \hspace{1em} \aif $\istee(i)$ \athen $\phi^\prime := \teestore(\phi_{prec})$
\State  \hspace{1em} \aelse
\State  \hspace{2em} \textbf{abort if} $\neg(\validqc (\phi_{prec}) \land view = v$
\Statex \hspace{7em} $\land phase = \com \land h=\hash(b))$
\State  \hspace{2em} $(lockv,lockh) := (v,h)$
\State  \hspace{2em} $\phi^\prime := \langle  \hash(b), view, \com \rangle_{\sigma_i}$
\State  \hspace{1em} \aendif
\State  \hspace{1em} send $\phi^\prime$ to leader; $(view, phase)++$

\State
\State  \Comment{decide phase}
\State  \textbf{as a leader} 
\State  \hspace{1em} wait for $\vec{\phi}$ s.t. \tmatch $(\vec{\phi}, Q_T, h, v, h^{\prime}, v^{\prime}, \com) $
\Statex \hspace{6em} $ \lor ~ \mmatch (\vec{\phi}, Q_M, h, v, h^{\prime}, v^{\prime}, \com)$
\State  \hspace{1em} \priobroadcast ( $qc_{com}  :=$ \combine($\vec{\phi}$))
\State  \textbf{all replicas}
\State  \hspace{1em} wait for $\langle h, v, h^{\prime}, v^{\prime}, \com \rangle_{\vec{\sigma}}$ from the leader
\State  \hspace{1em} \textbf{abort if} $\neg(\validqc (\langle  h, v, h^{\prime}, v^{\prime}, \com \rangle_{\vec{\sigma}}))$
\State  \hspace{1em} execute $b$ corresponding to $h$ and reply to clients

\State 
\State  \Comment{new-view phase}
\State  \textbf{upon timeout}
\State  \hspace{1em} $view$++; $phase=\nv$
\State  \hspace{1em} \aif $\istee(i)$ \athen $\phi_{nv} = \teeview(view)$ 
\State  \hspace{1em} \aelse $\phi_{nv} := \langle \perp, view, qc_{prep}._{Hprep}, qc_{prep}._{Vprep}, \nv \rangle_{\sigma_i}$
\State  \hspace{1em} \aendif
\State  \hspace{1em} send $\phi_{nv}$ to $view$'s leader

\end{algorithmic}
\end{minipage}
\end{algorithm*}

\subsection{Trusted Components} \label{subsec:Trustcompon}
In \sysname, TEE replicas invoke trusted logic inside TEEs through a trusted component called \counter. 
The role of \counter is to maintain trusted protocol state and enforce consistency constraints across certified protocol steps. 
In particular, \counter binds certified messages to monotonically increasing protocol identifiers and prevents conflicting certified protocol states from being generated by the same TEE-enabled replica.

\sysname builds upon the trusted components originally used in \damysus, namely \textsc{Checker} and \textsc{Accumulator}). 
At a high level, \textsc{Checker} validates proposal safety and binds certificates to monotonically increasing $(view, phase)$ identifiers, while \textsc{Accumulator} tracks the latest prepared state carried in leader-transition messages. 
In \sysname, these functionalities are tightly coupled in the protocol execution path. 
Therefore, we consolidate them into a single trusted module, called \counter, and adapt it to support the dual-quorum construction and leader-dependent fast paths. 
This consolidation reduces redundant enclave invocations and minimizes context switches between trusted and untrusted execution.

\counter maintains two trusted protocol states: the latest prepared block and the latest locked block. 
The locked block (\ie, the highest block for which the replica holds a valid QC) preserves safety by preventing conflicting executions, while the prepared block supports liveness by allowing replicas to relay their latest safe state during leader transitions.
By integrating the functionality of \textsc{Accumulator}, \counter also enables a TEE-enabled leader to safely select and extend the highest prepared block carried in leader-transition messages.

\bheading{\counter\ state.}
The state maintained by replica $p_i$'s \\  \counter consists of three components:

\begin{packeditemize}
   \item $\{sk_i, pk_1, \ldots, pk_n\}$, where $sk_i$ is the enclave-protected private key of replica $p_i$, and $\{pk_1, \ldots, pk_n\}$ are public keys;

   \item $(view, phase)$, where $view$ is the current protocol view and $phase$ is the current protocol phase. 
   This identifier monotonically increases whenever \counter generates a certified protocol message;

   \item $(prepv, preph)$ and $(lockv, lockh)$, which record the view number and hash of the latest prepared block and latest locked block, respectively.

   \item $(propv,proph)$, which record the view and hash of the last certified proposal, preventing the trusted component from certifying more than one proposal per view.
\end{packeditemize}

\bheading{\counter\ interface.}
The \counter component exposes three main trusted operations:

\begin{packeditemize}
    \item \teeprepare($b,h,\phi_{nv},\vec{\phi}$):
    takes a proposed block $b$ with hash $h$, a selected new-view certificate $\phi_{nv}$, and a set of new-view certificates $\vec{\phi}$.
    It validates the set, checks that $\phi_{nv}$ carries the highest valid prepared state, and verifies that $b$ safely extends that state.
    It then records and certifies the proposal, allowing at most one certification per view without advancing the protocol identifier.

    \item \teevote($b,\phi_{prep}$):
    takes a proposed block $b$, its proposal certificate
    $\phi_{prep}$.
    It validates the proposal using \safeproposal\ and,
    under a TEE-enabled leader, records the block as prepared.
    It then generates a certified vote and advances the
    local protocol identifier.

    \item \teestore($\phi$): 
    takes a certified protocol message $\phi$, verifies its validity, and updates the trusted prepared or locked state according to the protocol step associated with $\phi$. 
    It then outputs a certified confirmation of the state update.

\end{packeditemize}

\begin{algorithm}[!t]

\caption{TEE code for operations}
\label{alg:tee}
  \begin{algorithmic}[1]
  \State $(sk, pks)$ \Comment{private and public key}
  \State $(view, phase) = (0, \nv)$  \Comment{current view and phase}
  \State $(propv, proph) := (0,\perp)$ \Comment{last certified proposal}
  \State $(prepv, preph)=(0, H(\mathcal{G}))$  \Comment{latest prepared block}
  \State $(lockv, lockh)=(0, H(\mathcal{G}))$  \Comment{latest locked block}
  
  
  \State
  \State \textbf{function} \teeprepare$(b, h, \phi_{nv}, \vec{\phi})$
  \State \hspace{1em} \textbf{abort if}
$\neg(i=\leader(view) \land phase=\prep)$
  \State \hspace{1em} \textbf{abort if} $\neg\left(\begin{array}{l} 
       \validqc(\vec{\phi}) \wedge~\phi_{nv} \in \vec{\phi} \wedge\\
       \phi_{nv} = \arg\max_{\phi\in\vec{\phi}} \phi._{Vjust} )
    \end{array}\right)$

  \State \hspace{1em} $\langle \perp, v, h^{\prime}, v^{\prime},ph\rangle_{\sigma} := \phi_{nv}$
  \State \hspace{1em} \textbf{abort if} $\neg (v = view \wedge ph = \nv)$ 
  \State \hspace{1em} \textbf{abort if} $\neg ( \hash(b) = h \wedge b._{parent} = h^{\prime})$ 
  \State \hspace{1em} \textbf{abort if} $\neg ( h^{\prime} = lockh \lor v^{\prime} > lockv)$
  \State \hspace{1em} \textbf{abort if} $propv=view \land proph \neq \perp$
  \State \hspace{1em} $(propv,proph) := (view,h)$
  \State \hspace{1em} \return $\phi := \langle  \hash(b),  v, \prep\rangle_{\sigma}$

  \State 
  \State \textbf{function} \teestore$(\phi)$ 
  \State \hspace{1em} \textbf{abort if} $\neg (\validqc (\phi))$
  \State \hspace{1em} $\langle h, v, h^{\prime}, v^{\prime},ph\rangle_{\sigma} := \phi$
  \State \hspace{1em} \textbf{abort if} $\neg \big(v=view \land ((ph=\prep \land phase=\preco) $
  \State \hspace{10em} $\lor (ph=\preco \land phase=\com)) \big)$
  \State \hspace{1em} $prepv = v$; $preph = h$
  \State \hspace{1em} \aif $ph = \preco$ \athen 
  \State \hspace{2em} $lockv=v$; $lockh=h$
  \State \hspace{1em} \aendif
  \State \hspace{1em} $\phi := \langle  h, view, h^{\prime}, v^{\prime}, phase \rangle_{\sigma}$
  \State \hspace{1em} $(view, phase)++$ \Comment{increased to avoid  equivocation}
  \State \hspace{1em} \return $\phi$

  \State 
  \State \textbf{function} \teevote$(b,\phi_{prep})$
    \State \hspace{1em} \textbf{abort if} $\neg(phase=\prep$
    \State \hspace{6em}$\land \safeproposal(b,\phi_{prep}))$
    \State \hspace{1em} $(h',v') := (\phi_{prep}._{Hjust},\phi_{prep}._{Vjust})$
    \State \hspace{1em} $ph := \prep$
    \State \hspace{1em} \aif $\istee(\leader(view))$ \athen
    \State \hspace{2em} $(prepv,preph) := (view,\hash(b))$; $ph := \com$
    \State \hspace{1em} \aendif
    \State \hspace{1em} $\phi := \langle H(b),view,h',v',ph\rangle_{\sigma_i}$
    \State \hspace{1em} $(view,phase)++$
    \State \hspace{1em} \return $\phi$

  \State
    \State \textbf{function} \teeview$(w)$
    \State \hspace{1em} \textbf{abort if} $\neg(w>view \lor (w=view \land phase=\nv))$
    \State \hspace{1em} $(view,phase) := (w,\nv)$
    \State \hspace{1em}
    \return $\phi := \langle \perp, view, preph, prepv, phase\rangle_{\sigma_i}$

  \label{algo}
  \end{algorithmic}

\end{algorithm}

\subsection{The Algorithm}\label{subsec:fast}
Algorithm~\ref{alg:normal-op} and Algorithm~\ref{alg:tee} present the pseudocode of replica operations, where Algorithm~\ref{alg:tee} is executed only by TEE replicas. 
For brevity, the pseudocode shows only the standard view-change
path; the fast view-change optimization is omitted.
At a high level, \sysname follows two execution paths depending on the leader type. 
Under a TEE-enabled leader, the protocol leverages trusted state and certified proposal consistency enforced by \counter\ to shorten the commit path. 
Under a non-TEE leader, the protocol falls back to a standard BFT-style quorum path while preserving safety through quorum intersection.

\bheading{View change.}
When entering a new view, replicas send new-view certificates carrying their latest prepared state to the new leader (Algorithm~\ref{alg:normal-op}, lines 77 and 78). In the standard path, the new leader collects a valid dual quorum of new-view certificates, selects the certificate with the highest prepared view, and proposes a block extending the corresponding prepared state.

When the preceding view was led by a TEE-enabled leader and a valid TEE-certified proposal from that view is available, the proposal already certifies a safe extension derived from a valid new-view quorum. 
The subsequent leader can therefore use this certified state as its safe extension point without reconstructing the same safety information from another full set of view-change messages. 
If no such certificate is available, or if the preceding leader was non-TEE, the protocol follows the standard dual-quorum view-change procedure above.

\bheading{TEE-leader fast path.}
When the leader is TEE-enabled, it invokes \teeprepare\ to generate a certified proposal extending the highest prepared block carried in the new-view certificates. 
The trusted component binds the proposal to the current view and prevents conflicting certified proposals from being generated by the same TEE-enabled leader (Algorithm~\ref{alg:tee}, lines 15).

The certified proposal 
$\langle \phi_{nv}, b, \langle \hash(b), view, \prep \rangle_{\sigma} \rangle$
is then broadcast to replicas. 
Upon receiving the proposal, replicas verify that the proposed block safely extends the justified block carried in the certificate. 
TEE replicas invoke \teevote\ to generate certified responses, while non-TEE replicas validate proposals using \safeproposal\ and return signed votes once pass the check.

Since certified proposals generated by the TEE-enabled leader are uniquely constrained by \counter, the protocol can safely use a shortened commit path under this execution mode. 
Thus, backups directly generate commit certificates without proceeding through the intermediate \precommit and \commit phases used under non-TEE leaders (Algorithm~\ref{alg:normal-op}, lines 27). 
Once the leader collects a valid \quorum\ of commit certificates, it combines them into a commit QC and broadcasts the QC to all replicas. 
Upon receiving the commit QC, replicas execute the block.

\bheading{Non-TEE leader fallback path.}
When the leader is not TEE-enabled, \sysname follows a standard HotStuff-style commit path. 
The leader first broadcasts a proposal extending the highest prepared block collected during the leader transition. 
Replicas validate the proposal and generate prepare responses according to their trust configuration (Algorithm~\ref{alg:normal-op}, lines 22 and 28). 
TEE replicas invoke \teevote\ to generate certified prepare responses, while non-TEE replicas validate proposals using \safeproposal\ before issuing signed votes.
The remaining protocol phases follow the same quorum-collection pattern. 
The leader collects a valid quorum of certificates, combines them into a quorum certificate, and broadcasts the QC to replicas. 
TEE replicas update their trusted prepared or locked states through \teestore (Algorithm~\ref{alg:tee}, lines 25), while non-TEE replicas generate signed protocol votes locally. 
This process repeats across the \prepare, \precommit, and \commit phases until the leader gathers a valid commit quorum certificate and broadcasts the final commit certificate for execution.

\bheading{Dual-quorum processing.}
Throughout the protocol, quorum formation follows the dual-quorum construction introduced in \ssecref{subsec:dualqc}. 
A valid quorum certificate may therefore be formed either from a TEE-Quorum consisting exclusively of TEE replicas or from a Mixed-Quorum consisting of arbitrary replicas. 

When sufficient TEE replicas are available, leaders prioritize collecting TEE-backed certificates to accelerate quorum formation. 
To maximize the likelihood of forming a TEE-Quorum, leaders invoke \priobroadcast(msg), which opportunistically disseminates proposals to TEE replicas before extending the broadcast to all replicas.

\subsection{Correctness Analysis}

We establish correctness by first showing that any two valid quorums share a non-equivocating replica. We prove Prepared-Block Uniqueness and Safe Extension
Across Views in Appendix~C of the
\href{https://github.com/1wenwen1/Raftel/blob/main/doc/Breaking_Fault_Lines__Unifying_BFT_Consensus_in_a_Partially_Trusted_World.pdf}{extended version}.
Together, these lemmas imply safety. Complete safety and liveness proofs are also provided there.

\begin{lemma}[Quorum Intersection]\label{lem:qc} Suppose $n \ge 3f+1$ or $m \ge 2f+1$, with $Q_T = \max\{\lfloor m/2\rfloor+1, f+1\}$ and $Q_M=n-f$. Any two valid quorums, each either a TEE-Quorum or a Mixed-Quorum, intersect in at least one honest replica or one TEE replica. Such a replica cannot vote for conflicting blocks in the same view and phase. \end{lemma} 

\begin{proof} Let $S$ be the set of all $n$ replicas and $S_T\subseteq S$ the set of $m$ TEE replicas. We consider three cases. 
\noindent\textbf{Mixed/Mixed.} For two Mixed-Quorums $M_1,M_2\subseteq S$, \[ |M_1\cap M_2|\ge 2Q_M-n=n-2f. \] If $n\ge 3f+1$, the intersection contains at least $f+1$ replicas, including an honest one. Otherwise, $m\ge 2f+1$, and the number of non-TEE replicas satisfies $n-m\le n-2f-1$. Thus, the intersection is larger than the entire non-TEE set and contains a TEE replica. 
\noindent\textbf{TEE/Mixed.} For a TEE-Quorum $T\subseteq S_T$ and a Mixed-Quorum $M\subseteq S$, \[ |T\cap M|\ge Q_T+Q_M-n\ge (f+1)-f=1. \] Every replica in this intersection is TEE-enabled.  
\noindent\textbf{TEE/TEE.} For two TEE-Quorums $T_1,T_2\subseteq S_T$, \[ |T_1\cap T_2|\ge 2Q_T-m \ge 2\lfloor m/2\rfloor+2-m\ge 1. \] Their intersection therefore also contains a TEE replica. In all cases, an honest replica follows the voting rule, while a TEE replica enforces non-equivocation even if its host is Byzantine. Neither can vote for conflicting blocks in the same view and phase. 
\end{proof}

\begin{theorem}[Safety]\label{lem:safety} No two honest replicas execute conflicting blocks. \end{theorem} \begin{proof}[Proof sketch] By Prepared-Block Uniqueness, conflicting blocks cannot both become prepared in the same view. By Safe Extension Across Views, once a block becomes prepared, any subsequently prepared block must extend it.  Since an honest replica executes a block only after the commit condition is satisfied, two honest replicas cannot execute conflicting blocks. \end{proof}

\section{\csysname Design}

\subsection{Overview} 
\csysname extends \sysname with a chained pipeline structure inspired by Chained-HotStuff~\cite{HotStuffYin2019} and \damysus~\cite{damysus}. 
Instead of completing all protocol phases for a block within a single view, \csysname pipelines justification and commitment across consecutive views. 
As a result, leaders of adjacent views overlap in responsibility: each leader both proposes a new block and helps justify or commit blocks proposed in earlier views.

At a high level, blocks carry quorum certificates that justify their parent blocks. 
When a leader proposes a new block, the QC embedded in the proposal simultaneously certifies the predecessor block. 
This pipelined structure allows commitment decisions to propagate continuously along the chain while reducing coordination overhead across views.

Most importantly, hybrid trust introduces asymmetric commit depth in the chained pipeline. 
Blocks proposed by TEE-enabled leaders can be committed after forming a one-chain, while blocks proposed by non-TEE leaders require a three-chain to preserve safety. 
As a result, the commit latency of a block depends on the trust guarantees of the leader who proposed it.

\bheading{Fast commit.} The chained structure also enables cross-view acceleration. 
When a non-TEE leader’s block is immediately followed by a block proposed by a TEE-enabled leader, the TEE-backed proposal can safely accelerate commitment of its predecessor. 
For example, in \figref{fig:chained}, when block $b_2$ is committed, its predecessor $b_1$ is simultaneously committed without waiting for the commit phase of the next view. 
Consequently, some non-TEE blocks may also be finalized earlier under mixed-trust leader sequences.
If a leader fails to gather sufficient votes, replicas fall back to standard new-view processing. 
Replicas send new-view messages carrying their latest prepared certificates, allowing the next leader to safely reconstruct the highest prepared chain and continue progress.

\begin{figure}[t]
    \centering
    \includegraphics[width=8cm]{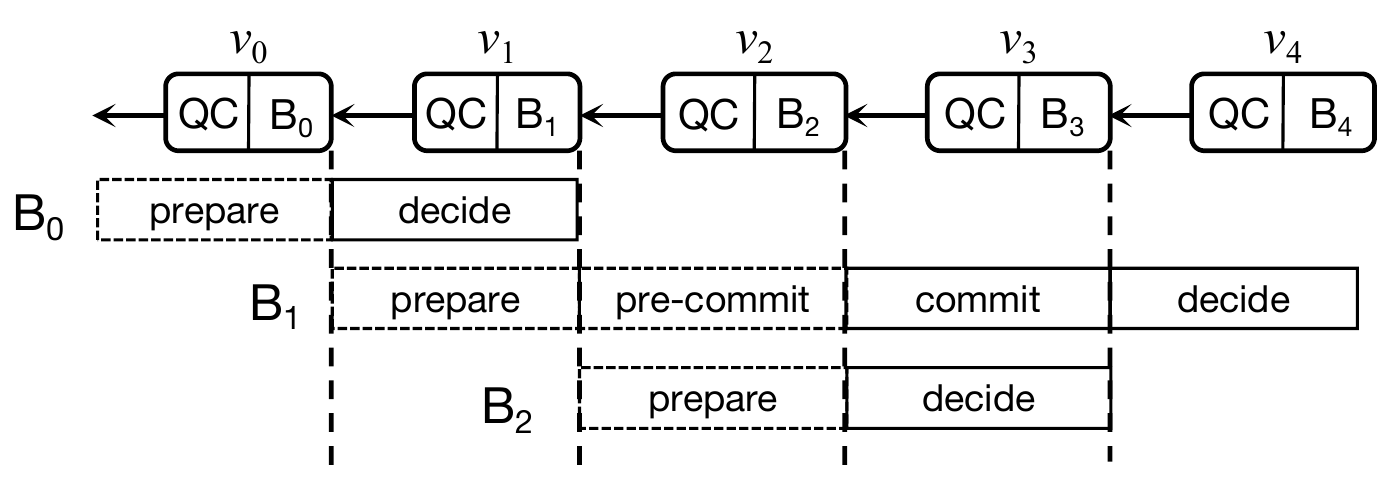}
    \caption{\csysname allows replicas to pre-commit the predecessor block $b_2$ of a block $b_3$ when preparing $b_3$, while simultaneously committing $b_2$'s predecessor block $b_1$.}
    \Description{Block-chain diagram and aligned phase timelines for Chained-Raftel across views v0 through v4. Blocks B0 through B4 each contain a quorum certificate and point to their predecessor. B0 is prepared in v0 and decided in v1. B1 progresses through prepare, pre-commit, commit, and decide in views v1 through v4. B2 is prepared in v2 and decided in v3, overlapping with the commit phase of B1. The timelines illustrate how progress on successive blocks is pipelined across views.}
    \label{fig:chained}

\end{figure}

\subsection{Definitions}
We now detail how we adapt some of the concepts introduced previously to handle a chained version.

\bheading{Blocks and quorum certificates.}
As in Chained-HotStuff\\~\cite{HotStuffYin2019}, blocks are organized into a chained structure. 
We replace \creatleaf\ with \creatchain, which generates chained blocks while filling missing gaps in the chain if necessary.
Each block stores:
(1) a hash pointer to its parent block, denoted by $b._{parent}$; and
(2) a quorum certificate stored in $b._{just}$, which justifies the parent block.

A quorum certificate has the form
$qc=\langle v,h\rangle_{\vec{\sigma}}$,
where $(v,h)$ identify the justified block and $\vec{\sigma}$ is the set of partial signatures collected for that block. 
As in \sysname, quorum formation follows the dual-quorum construction and may therefore use either TEE-Quorum or Mixed-Quorum certificates.

\bheading{Pipeline structure.}
The chained protocol contains only two protocol phases: \prepare and \newview, identified by tags \prep\ and \nv, respectively. 
Leaders continuously pipeline proposals across views, allowing justification and commitment to overlap between consecutive blocks.
A block forms a \textit{one-chain} if it directly extends the block referenced by its QC. 
Similarly, a block forms a \textit{three-chain} if three consecutive parent-linked blocks carry consecutive justifications.

\subsection{Commit Semantics}

The commit rule in \csysname depends on both the chain structure and the trust guarantees of the leader.

\bheading{TEE-enabled leaders.}
If a block proposed by a TEE-enabled leader forms a one-chain, the block and all of its ancestors are committed immediately. 
This optimization is enabled by the trusted proposal consistency enforced by \counter, which prevents conflicting certified proposals from being generated for the same protocol step.

\bheading{Non-TEE leaders.}
If a block proposed by a non-TEE leader forms a three-chain, the head of the three-chain and all of its ancestors are committed. 
This rule follows the standard chained BFT commit logic and preserves safety through quorum intersection across consecutive views.

\bheading{Mixed-trust acceleration.}
The chained pipeline also enables mixed-trust commit acceleration. 
When a non-TEE leader’s block is immediately followed by a block proposed by a TEE-enabled leader, the TEE-backed proposal may safely accelerate commitment of its predecessor. 
As a result, some non-TEE blocks can be finalized earlier than under the standard three-chain rule without compromising safety.

\subsection{The Algorithm}

The protocol follows a pipelined \prepare--\newview workflow in which consecutive leaders overlap in both proposal generation and predecessor justification.

\bheading{Prepare phase.}
At the beginning of each view, the leader constructs a new chained block extending the block certified by the highest prepared QC currently known to the replica. 
If the leader already holds the latest prepared QC from the previous view, it directly extends that chain. 
Otherwise, it reconstructs the latest prepared chain from new-view messages.

The leader then broadcasts the proposal together with its justification certificate. 
TEE-enabled leaders invoke \teeprepare\ to generate certified proposals, while non-TEE leaders attach signed prepare messages. 
Replicas validate the proposal chain and return prepare responses according to their trust configuration.

\bheading{New-view phase.}
At the end of a view, replicas forward their latest prepared certificates to the next leader. 
If the current leader fails to gather a valid quorum, replicas increment their local views and continue protocol execution using the highest prepared certificate carried in new-view messages.

\bheading{Pipelined commitment.}
Unlike in \sysname, commitment decisions are not completed within a single view. 
Instead, quorum certificates carried by newly proposed blocks continuously justify and finalize predecessor blocks along the chain. 
As the chain grows, replicas evaluate the one-chain or three-chain commit conditions according to the trust guarantees of the corresponding leaders.

\subsection{Correctness Analysis}

The proof of \csysname's safety and liveness properties is similar to \sysname's, except for the differences introduced by its pipelining structure. In particular, pipelining allows multiple blocks to be concurrently in flight, but does not change the quorum intersection properties or the non-equivocation guarantees on which the correctness arguments rely. The safety argument therefore extends by reasoning about the ordering and dependencies among consecutive pipelined blocks, while the liveness argument accounts for their asynchronous progression across views. Due to space constraints, we leave the detailed proof of safety and liveness to Appendix~C of the \href{https://github.com/1wenwen1/Raftel/blob/main/doc/Breaking_Fault_Lines__Unifying_BFT_Consensus_in_a_Partially_Trusted_World.pdf}{extended version}.

\section{Performance Evaluation} \label{sec:evaluation}

We implement a prototype of \sysname and \csysname\footnote{
Code available at \url{https://github.com/1wenwen1/Raftel}.
}, and evaluate their performance in both LAN and WAN. 
We compare \sysname and \csysname with \hotstuff, Basic-\damysus (referred to as \damysus in the evaluation), and \achilles to show the performance improvements. 
We aim to answer the following questions:

\begin{packeditemize}
    \item \textbf{Q1:} How does \sysname perform under varying fault thresholds in WAN compared to its counterparts?  
    
    \item \textbf{Q2:} How does \sysname perform under different TEE deployment? 
    
    \item \textbf{Q3:} How much performance overhead is introduced by SGX-related operations? 
\end{packeditemize}

\subsection{System Implementation and Setup}\label{sec:imple}
\bheading{Implementation.} We use Intel SGX to provide trusted services and build our prototype on the \damysus implementation~\footnote{Available at https://github.com/vrahli/damysus.}. 
All protocols are implemented in C++. 
We build on Damysus’s trusted components, \textsc{Checker} and \textsc{Accumulator}, and consolidate their functionalities into a unified trusted module, \counter, which is adapted to support our dual-quorum construction and fast-path designs. 
The detailed design is presented in \ssecref{subsec:Trustcompon}.
We use the OpenSSL library~\cite{OpenSSL} to realize ECDSA signatures with the \textit{prime256v1} elliptic curves and Salticidae~\cite{salticidae} for replica connections. 
Following prior work~\cite{damysus, achilles}, we do not implement remote attestation in the prototype, as it is only needed during initialization and does not affect steady-state performance.

\begin{figure}[t]
    \centering
    \begin{subfigure}[b]{0.49\linewidth}
        \centering
        \includegraphics[width=\linewidth]{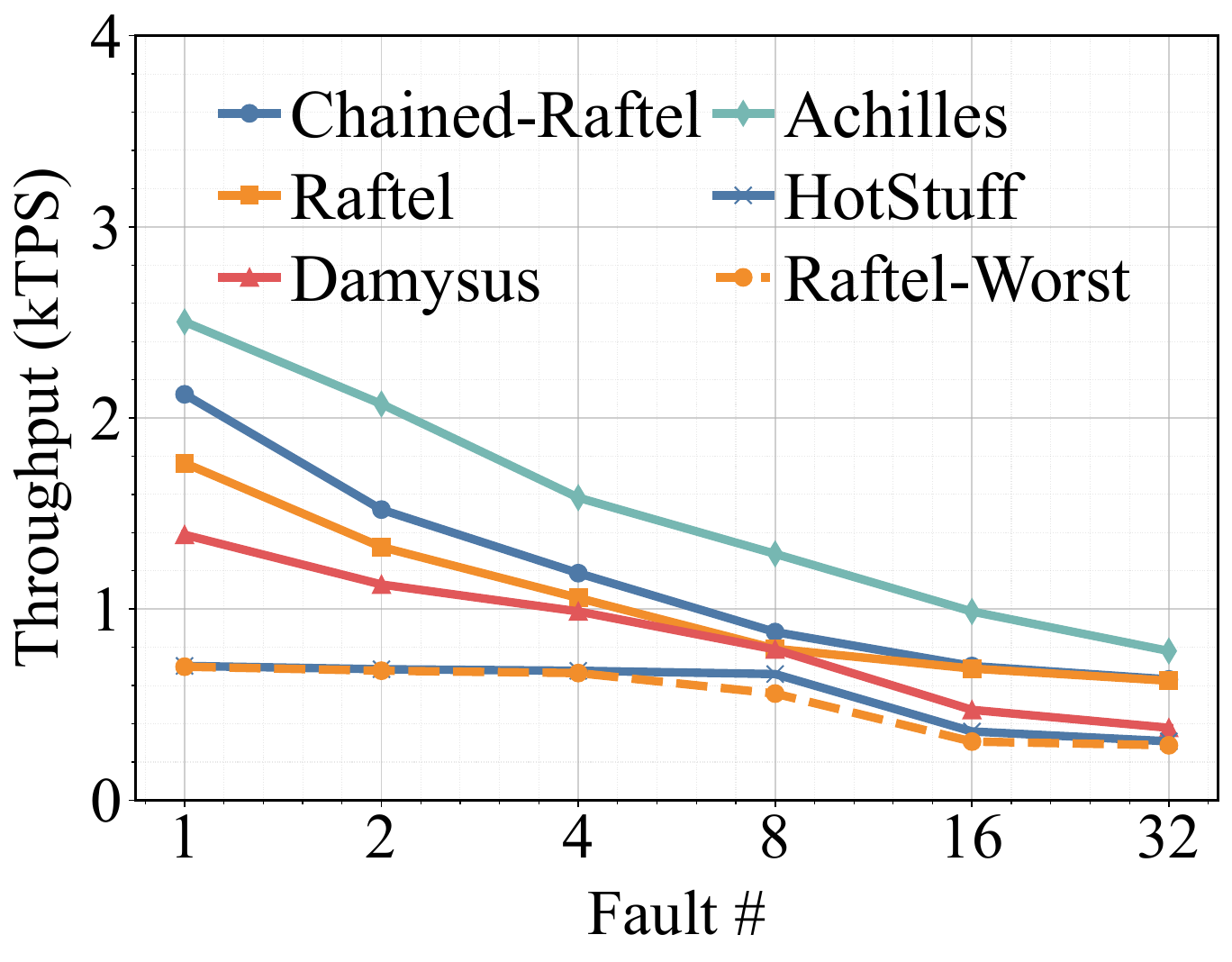}
        \caption{Throughput, WAN}
        \label{fig:t-WAN-replica}
    \end{subfigure}
    \hfill
    \begin{subfigure}[b]{0.49\linewidth}
        \centering
        \includegraphics[width=\linewidth]{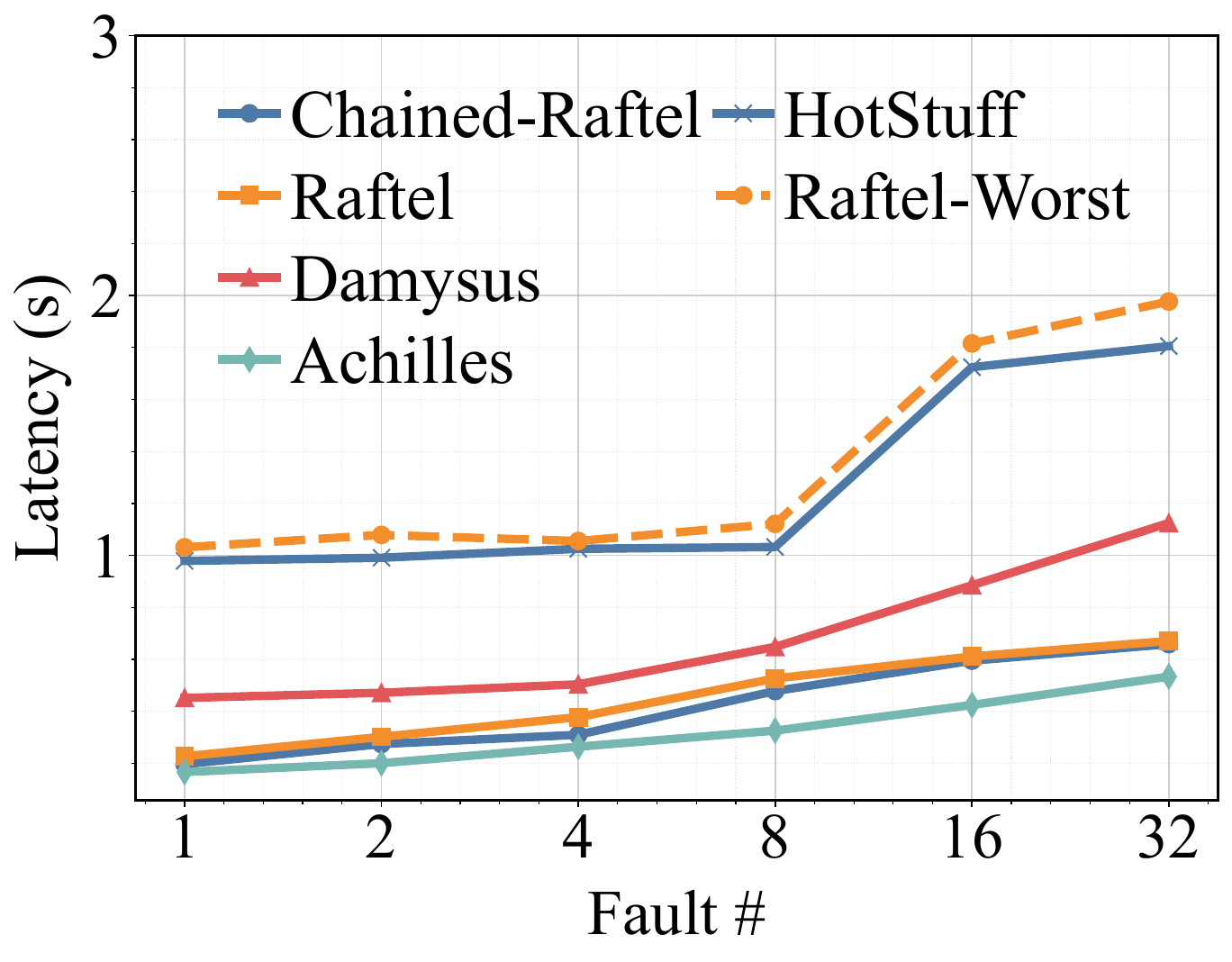}
        \caption{Latency, WAN}
        \label{fig:l-WAN-replica}
    \end{subfigure}

    \caption{Throughput and latency of \sysname with varying faults in WAN.}
    \Description{Two line charts compare six protocols in a WAN as the fault threshold increases from 1 to 32. The left panel shows throughput in thousands of transactions per second; the right shows latency in milliseconds. Throughput decreases and latency increases as the threshold grows. Achilles maintains the highest throughput and lowest latency, followed by Chained-Raftel and Raftel. Damysus has higher latency than these three protocols. Raftel-Worst closely tracks HotStuff in throughput and has slightly higher latency.}

    \label{fig:WAN-replica}
\end{figure}

\bheading{Experimental setup.}\label{sec:expset}
We conduct all experiments on a public cloud platform using up to 97 SGX-enabled instances, with each replica deployed on a dedicated virtual machine. Each VM is provisioned with 8 vCPUs, 32 GB of RAM, and runs Ubuntu Linux 20.04. All instances are connected through a private network interface with a bandwidth of 10 Gbps.
We evaluate two deployment scenarios: a local area network (LAN) and a wide area network (WAN). In the LAN setting, the average inter-replica RTT is $0.1 \pm 0.02$ ms. Because SGX-enabled instances are only available in a limited set of regions, we emulated WAN conditions using NetEm~\cite{netem}, configuring the inter-replica RTT to $100 \pm 5$ ms.

\bheading{Metrics.} We consider two types of performance metrics: (1) commit throughput measures the number of committed transactions per second (TPS), while commit latency measures the average delay from when a leader proposes transactions to when they are executed; (2) end-to-end throughput measures the number of client replies received per second (TPS), while end-to-end latency measures the average delay from when clients create transactions to when replies are received.
In most experiments, following \damysus, we use commit throughput/latency, since these metrics reduce the impact of client-side overheads and enable fairer protocol-level comparisons. We only use end-to-end throughput/latency in \figref{fig:redis} and \figref{fig:tvl} to evaluate overall system scalability and client-perceived performance.

\begin{figure}[t]
    \centering
    \begin{subfigure}[b]{0.49\linewidth}
        \centering
        \includegraphics[width=\linewidth]{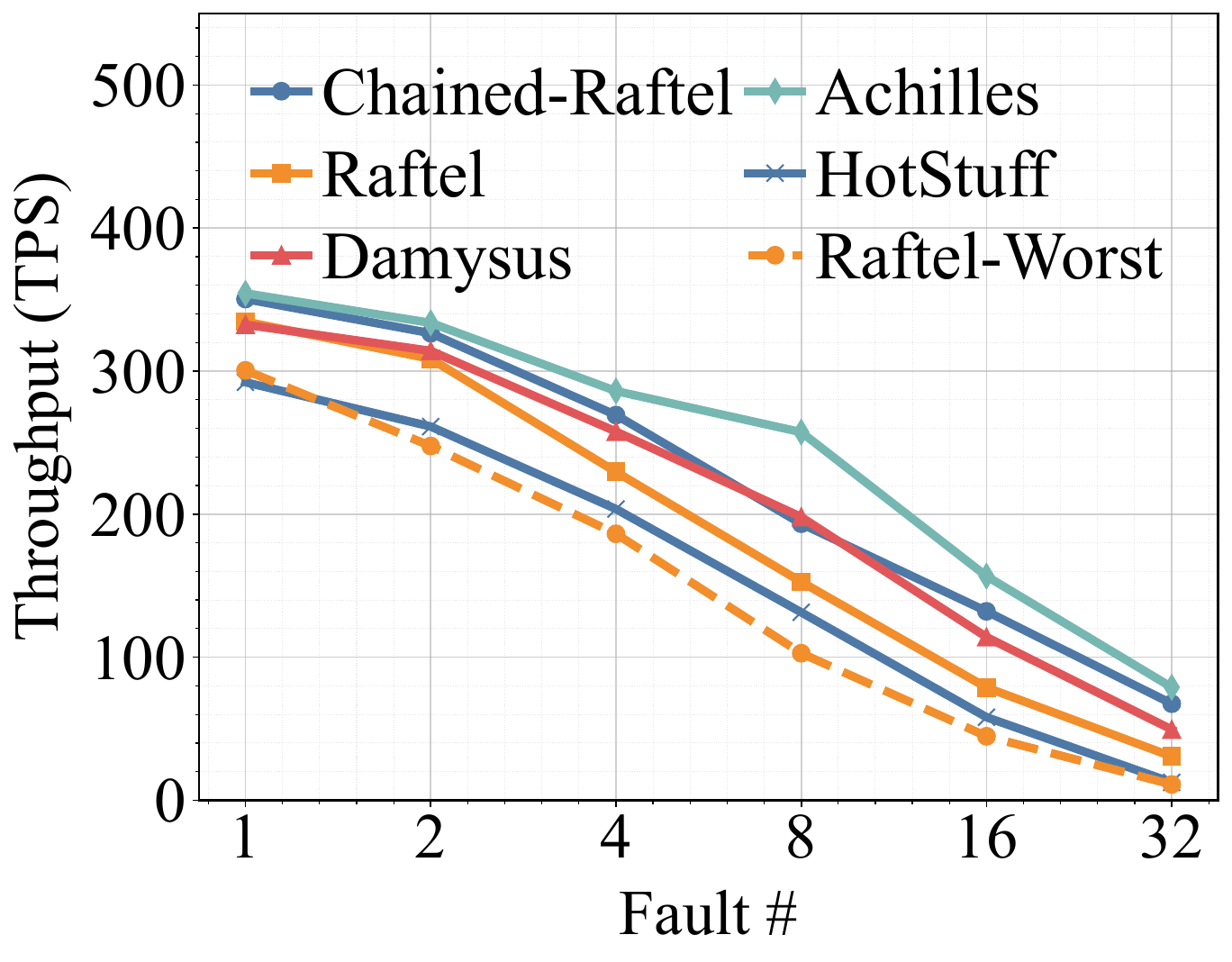}
        \caption{Throughput, LAN}
        \label{fig:t-LAN-replica}
    \end{subfigure}
    \hfill
    \begin{subfigure}[b]{0.49\linewidth}
        \centering
        \includegraphics[width=\linewidth]{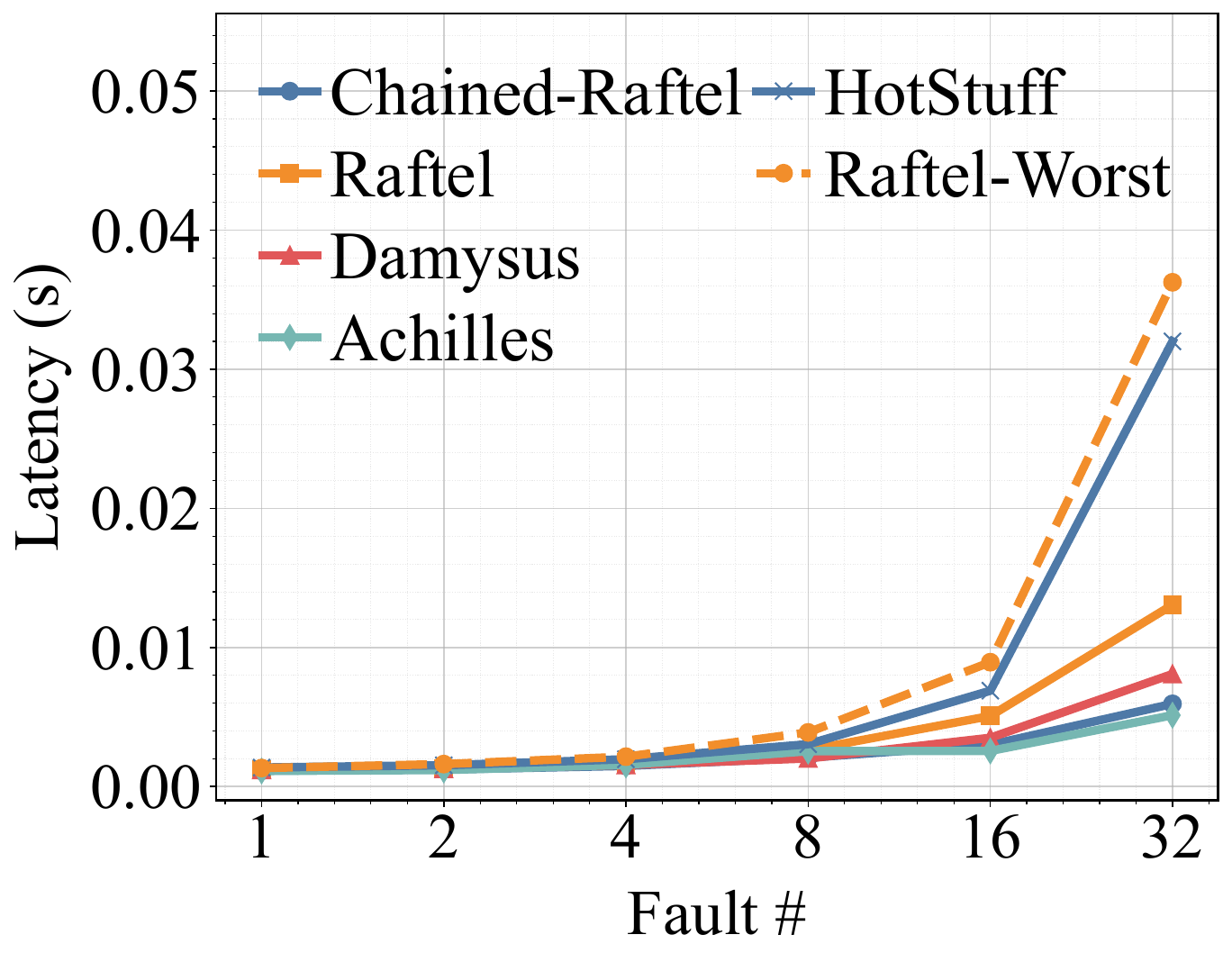}
        \caption{Latency, LAN}
        \label{fig:l-LAN-replica}
    \end{subfigure}

    \caption{Throughput and latency of \sysname with varying faults in LAN.}
    \Description{Two line charts compare HotStuff, Damysus, Raftel, Chained-Raftel, and Achilles in a LAN for fault thresholds from 1 to 32. The left panel plots throughput in thousands of transactions per second, which decreases for all protocols. The right plots latency in milliseconds, with differences widening at larger thresholds. At a threshold of 32, Achilles has the highest throughput and lowest latency, followed by Chained-Raftel; HotStuff has the lowest throughput and highest latency.}

    \label{fig:LAN-replica}

\end{figure}

\subsection{Scalability Evaluation} \label{sec:wholeperformane}
We evaluate the performance of \sysname and \csysname in the best case, \ie, when the number of TEE replicas is larger than $f$ to enable dual-quorum, and the leader is TEE-enabled to support the fast path, with the system parameters set to $m = f + 1$ and $n = 3f + 1$.
We also evaluate the performance of \sysname (referred to as \sysname-Worst) in the worst case, \ie, when the number of TEE replicas is zero with the system parameters set to $m = 0$ and $n = 3f + 1$.
We evaluate \sysname, \csysname, and its counterparts in WAN with varying fault thresholds $f \in$ \{1, 2, 4, 8, 16, 32\}. 
All protocols adopt 400 transactions per block and 256 B payload for each transaction. 

\subsubsection{Performance in WAN}
\figref{fig:t-WAN-replica} and \figref{fig:l-WAN-replica} show the throughput and latency of the six protocols under varying fault thresholds $f$ in a WAN setting. 
Since network communication is the dominant cost in WAN environments, the number of commit phases determines both latency and throughput due to serialized block generation. 
The throughput of \damysus and \hotstuff is lower than that of \sysname, with \hotstuff having the lowest throughput because it has the most commit phases.
\sysname-Worst exhibits performance comparable to \hotstuff, since both protocols have the same number of commit phases, while \sysname-Worst additionally introduces extra protocol-side decision logic.
\sysname and \csysname achieve up to $1.02\times$ and $1.05\times$ higher throughput than \hotstuff, with latency reduced by 61.46\% and 62.84\% respectively at $f=32$. 
However, \csysname lags behind \achilles, showing 18.9\% lower throughput and 21.2\% higher latency under the same conditions. 
This gap stems from the system size: \sysname requires $3f+1$ replicas, compared to $2f+1$ in \achilles. 
Although only $f+1$ votes are needed in \sysname, its leader still broadcasts proposals to $3f+1$ replicas, incurring higher communication overhead.

\subsubsection{Performance in LAN}
\figref{fig:t-LAN-replica} and \figref{fig:l-LAN-replica} report LAN throughput and latency. The overall trends are similar to those in WAN, except that \sysname generally achieves lower throughput and higher latency than \damysus. One possible explanation is that, with low network latency, the cost of involving more replicas and performing additional protocol checks outweighs \sysname's savings from fewer commit phases. The gap between \hotstuff/\damysus and the optimized protocols also narrows, as lower network latency reduces the cost of additional commit phases.
At $f=32$, \achilles achieves the highest throughput (79.0~kTPS) and lowest latency (5.13~ms). \csysname ranks second, reaching 67.5~kTPS and 5.95~ms. Compared with \damysus and \hotstuff, \csysname achieves $1.37\times$ and $5.40\times$ higher throughput, respectively, while reducing latency by 26.5\% and 81.4\%.

\begin{figure}
    \centering
    \includegraphics[width=0.9\linewidth]{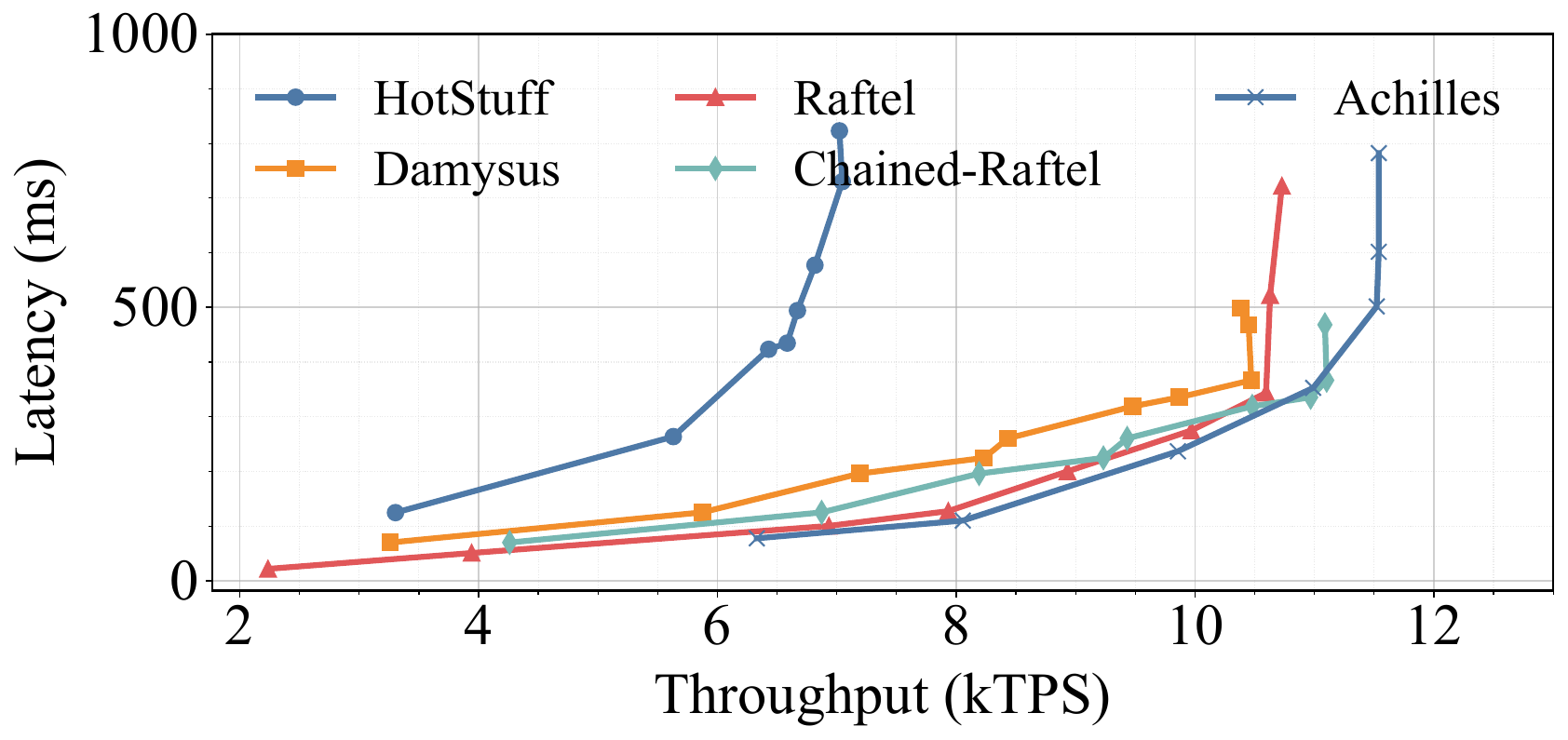}
    \caption{End-to-end Throughput vs. Latency in LAN.}
    \Description{Line chart of end-to-end latency in milliseconds versus throughput in thousands of transactions per second in a LAN, comparing five protocols as offered load increases. Latency rises sharply near saturation. HotStuff saturates at about 7 thousand transactions per second; Damysus at about 10.4 thousand; Raftel at about 10.7 thousand; Chained-Raftel at about 11.1 thousand; and Achilles at about 11.5 thousand. The experiment uses a fault threshold of 8, 256-byte payloads, and batches of 400 transactions.}

\label{fig:tvl}

\end{figure}

\subsubsection{Throughput vs. latency}
\figref{fig:tvl} illustrates the {end-to-end latency} (\ie, from when clients create transactions to when replies are received) and the corresponding  throughput of the five protocols as the offered load increases until system saturation.
The fault threshold $f$ is set to 8, the payload is 256 B, and the batch size is 400.
The results show that \sysname and \csysname  achieve maximum throughput of 10.7 kTPS and 11.1 kTPS, respectively.
\sysname achieves significantly better performance compared to \damysus and \hotstuff due to its optimal quorum size of $f+1$ and one-phase commit.
\hotstuff performs worse than \damysus, with a maximum throughput of 7.0 kTPS, since it requires $3f+1$ replicas and a quorum size of $2f+1$. 
\achilles shows the highest throughput with a maximum of 11.5 kTPS because of its fewer commit phases and the smallest replica count.

\subsection{Microbenchmark Evaluation} \label{subsec:ablation}
To evaluate the impact of performance acceleration caused by TEE replicas, we measure throughput and latency of \sysname under four scenarios:

\begin{packeditemize}
\item \textbf{\sysname-S1}: Leaders rotate among at least $f+1$ TEE replicas (with $m = f+1$) with both TEE-leader fast path and TEE-Quorum enabled (\ssecref{subsec:tee-leader}), yielding the best performance (as shown in prior experiments).

\item \textbf{\sysname-S2}: Leaders rotate among TEE replicas, but fewer than $f+1$ are available (\ie, $m = f$). In this case, only the TEE-leader fast path is enabled.

\item \textbf{\sysname-S3}: Leaders rotate among non-TEE replicas, while at least $f+1$ TEE replicas exist in the system (\ie, $m = f+1$). Here, only the TEE-Quorum is enabled.

\item \textbf{\sysname-S4}: Leaders rotate among non-TEE replicas, with fewer than $f+1$ TEE replicas (\ie, $m = f$). No optimization is enabled, and performance matches classical TEE-free BFT protocols.
\end{packeditemize}

\figref{appenfig:sets} illustrates the throughput and latency of \sysname under these four scenarios in LAN. 
The payload is 256 B, and the batch size is 400.
\sysname-S1 achieves the highest throughput of 31.5 kTPS with 32 faults. 
Performance degrades in \sysname-S2 and \sysname-S3, with \sysname-S2 slightly outperforming \sysname-S3, indicating that the cost of a larger quorum size is smaller than the cost of an additional consensus phase. 
\sysname-S4 is the worst case, yielding the lowest throughput. 
Overall, the performance gap across the four scenarios is not large, and the difference becomes even less noticeable when the system size is small.

\begin{figure}[t]
    \centering
    \begin{subfigure}[b]{0.49\linewidth}
        \centering
        \includegraphics[width=\linewidth]{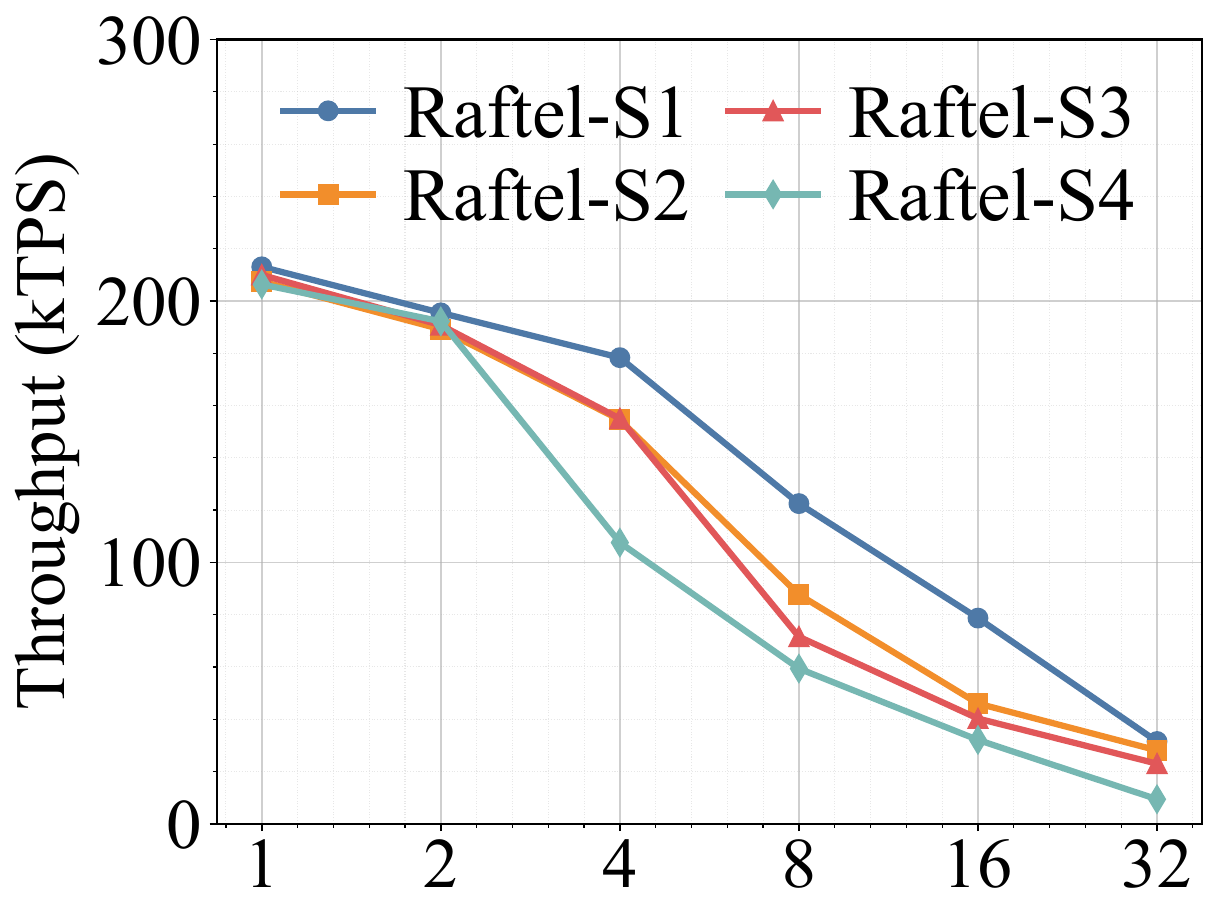}
        \caption{Throughput, LAN}
        \label{fig:sets}
    \end{subfigure}
    \begin{subfigure}[b]{0.49\linewidth}
        \centering
        \includegraphics[width=\linewidth]{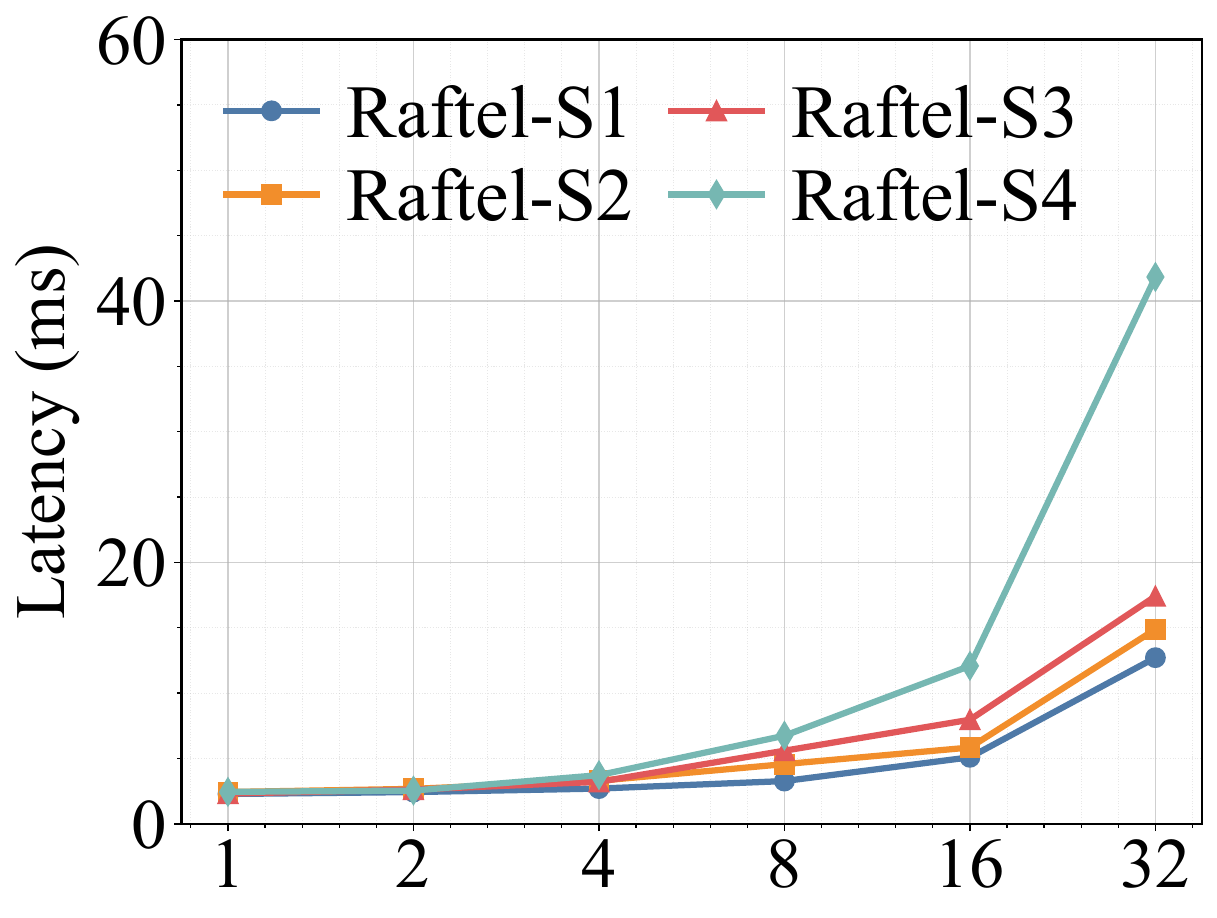}
        \caption{Latency, LAN}
        \label{fig:lcounter}
    \end{subfigure}
    \caption{\sysname's performance under different scenarios.}
    \Description{Two line charts show Raftel throughput in thousands of transactions per second and latency in milliseconds for fault thresholds from 1 to 32 in a LAN. S1 enables both the TEE-leader fast path and TEE quorums; S2 enables only the fast path; S3 enables only TEE quorums; S4 enables neither. At larger thresholds, S1 achieves the highest throughput and lowest latency, followed by S2, S3, and S4. Differences are small at the lowest thresholds.}

    \label{appenfig:sets}
\end{figure}

\subsection{Performance Under Faults}
We evaluate \sysname under crash faults in a WAN deployment with 25 replicas and a fault threshold of 8. We do not evaluate Byzantine equivocation faults, since \sysname handles equivocation in the same manner as conventional BFT protocols and introduces no additional protocol-level differences in this case.
We study the impact of crash faults on real-time throughput under two scenarios: crashes of the TEE leader and of a non-TEE leader. The view-change timeout is set to 2 seconds. \figref{fig:failure} shows the average committed throughput over time.
In both cases, the leader crashes at 10 seconds, causing the throughput to drop to nearly zero at 10.8 s. The view-change procedure is then triggered to recover from the failure. Recovery completes at 13.3 seconds for the TEE-leader crash and at 13.7 seconds for the non-TEE leader crash, after which throughput returns to the steady state.
The faster recovery in the TEE-leader case is enabled by \sysname's  TEE-leader fast view change.

\begin{figure}
    \centering
    \includegraphics[width=0.99\linewidth]{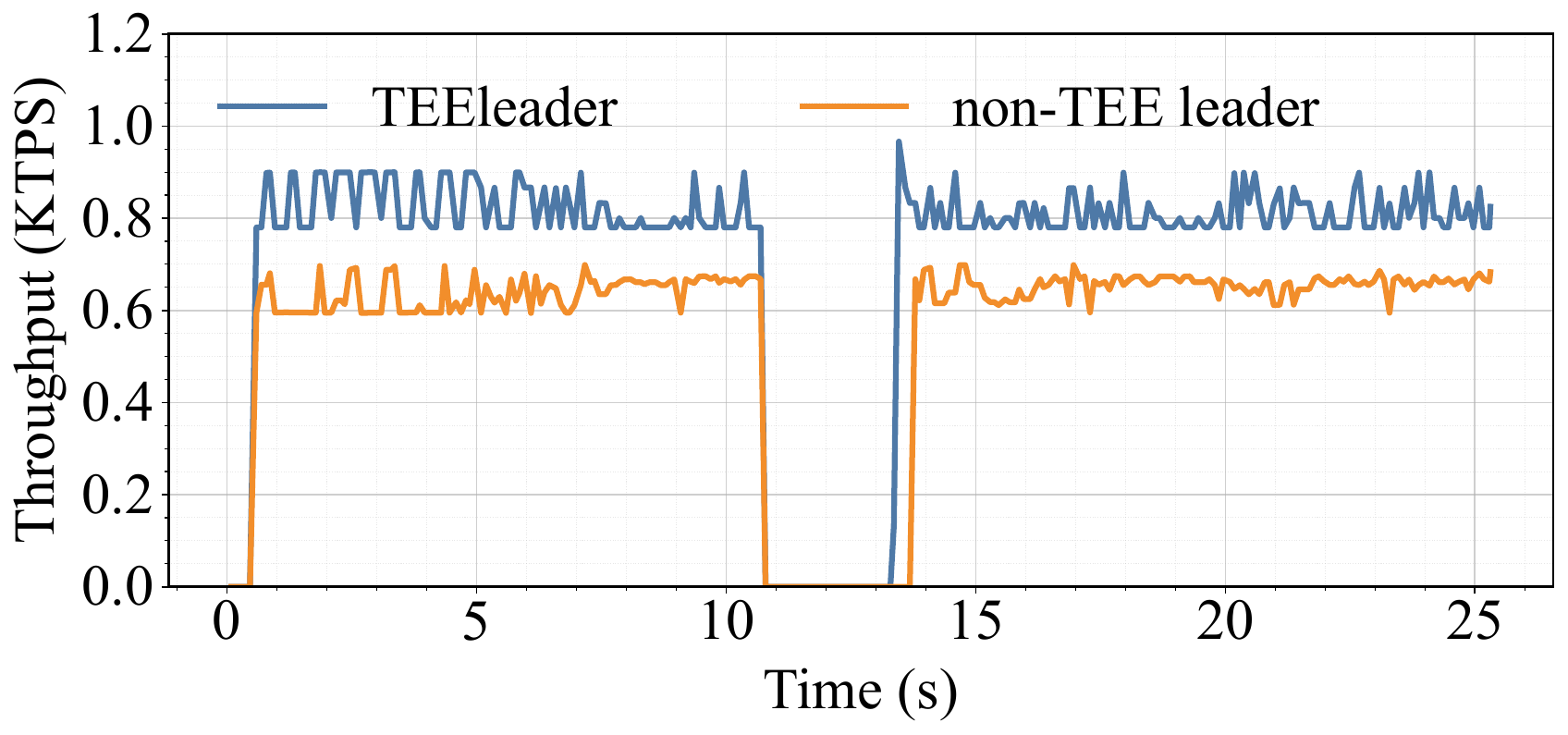}
    \caption{Performance of \sysname under faults.}
    \Description{Time-series chart of Raftel throughput in thousands of transactions per second over approximately 25 seconds in a WAN. Separate curves show crashes of a TEE leader and a non-TEE leader. Following a crash at 10 seconds, both curves fall to nearly zero around 10.8 seconds. Throughput resumes around 13.3 seconds for the TEE-leader case and 13.7 seconds for the non-TEE-leader case, returning to approximately 0.8 and 0.65 thousand transactions per second, respectively.}

\label{fig:failure}
\end{figure}
\begin{table}[t]
    \centering
    \renewcommand{\arraystretch}{1.1}
    \caption{Overhead profiling for \sysname in LAN.}
    \label{tab:overhead}
    \begin{tabular}{c|ccc|ccc}
    \toprule
     \multirow{2}{*}{Protocols} & \multicolumn{3}{c}{Throughput (kTPS)} & \multicolumn{3}{c}{Latency (\msecs)}  \\ 
        & $f=2$ & $f=4$ & $f=8$ & $f=2$ & $f=4$ & $f=8$\\
    \midrule
    \sysname & 192.4 &151.6 & 89.7 &2.1  &2.6  &4.5\\
    \sysname-C & 205.4 &158.5 &92.9 &2.0 &2.5 &4.3\\ 
    \bottomrule
\end{tabular}
\end{table}

\subsection{Overhead Profiling} \label{subsec:overhead}
To understand the overhead of using SGX, we implement a variant of \sysname, called \sysname-C, which operates the trusted components outside the SGX enclave. 
On the critical path, a phase of \sysname incurs 2 enclave interactions when the leader is TEE-enabled, and 4 enclave interactions when the leader is non-TEE, where each enclave interaction consists of one ecall and one ocall. 

Table~\ref{tab:overhead} shows the performance of \sysname and \sysname-C in LAN with varying fault thresholds $f \in$ \{2, 4, 8\}.
{Each value is averaged over 5 runs; we report one decimal place for readability.}
\sysname-C consistently outperforms \sysname, \eg, $205.4$ kTPS vs. 192.4 kTPS and $2.0$ ms vs. $2.1$ ms latency given $f=2$. 
The gap remains modest across all settings (3.4\% to 6.3\% throughput reduction and up to 0.2 ms latency increase), indicating that SGX introduces slight overhead.

\subsection{KV-Store Application}
We deploy Redis~\cite{redis} on top of the \sysname protocols and use a workload consisting of 100\% \texttt{SET} operations with 1 KB values.
The fault threshold $f$ is 8, and the batch size is 400.

\figref{fig:redis} illustrates the end-to-end latency (\ie, from when clients issue requests to when replies are received) and the corresponding throughput of the five protocols as the offered load increases until system saturation.
The results show that \sysname and \csysname  achieve maximum throughput of 84 TPS and 92 TPS, respectively.
\sysname significantly outperforms \damysus and \hotstuff due to its quorum size of $f+1$ and one-phase commit design.
\hotstuff achieves the lowest throughput (44 TPS) because it requires $3f+1$ replicas and a quorum size of $2f+1$.
\achilles achieves the highest throughput (95 TPS) due to its fewer commit phases and smaller committee size.
Overall, these results remain consistent with the microbenchmark results, showing that the protocol-level optimizations of \sysname effectively translate into end-to-end application performance gains.

\begin{figure}
    \centering
    \includegraphics[width=0.99\linewidth]{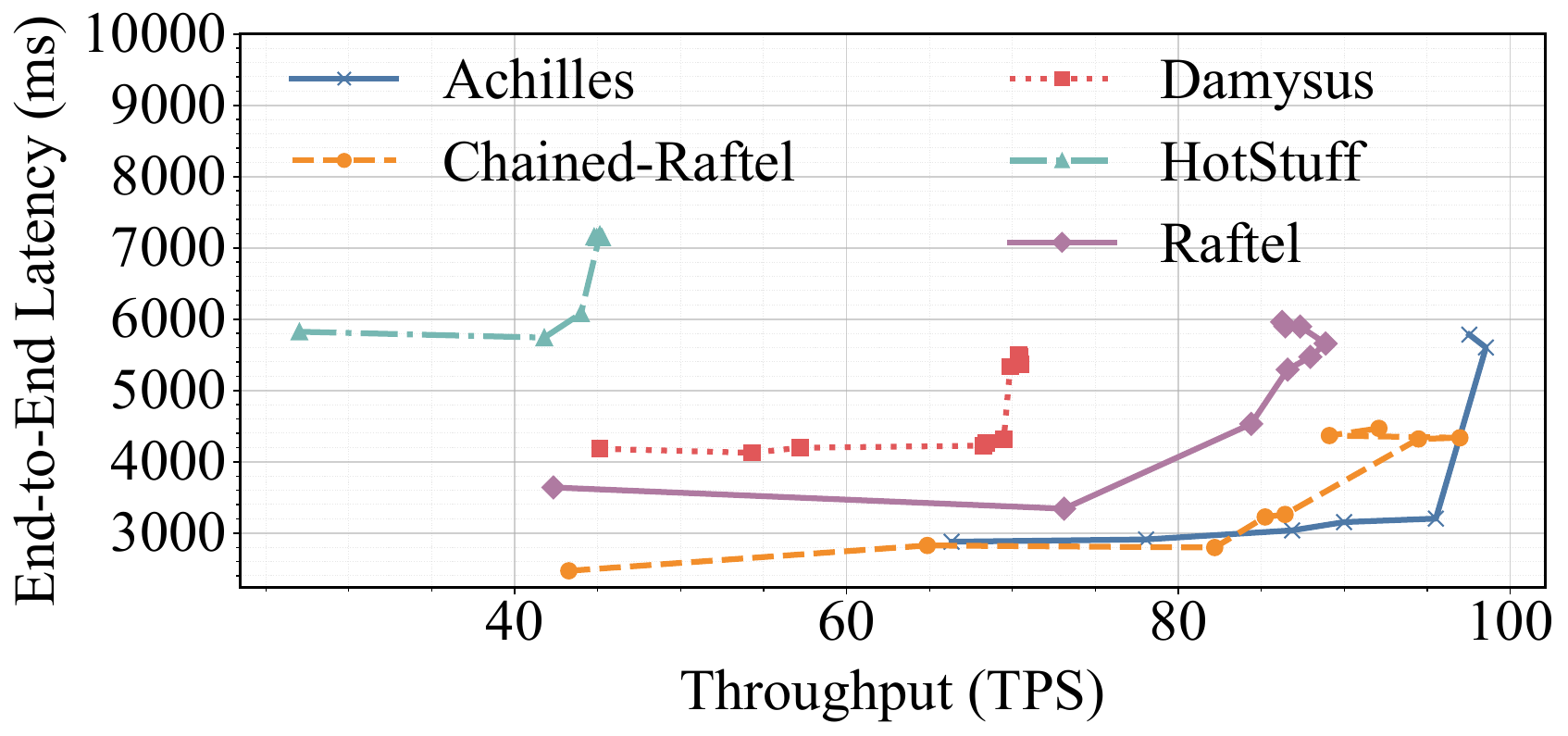}
    \caption{End-to-end Throughput vs. Latency.}
    \Description{Line chart comparing end-to-end latency in milliseconds against throughput in transactions per second for five protocols running a Redis workload of SET operations with 1-kilobyte values. Achilles and Chained-Raftel reach the highest throughputs, followed by Raftel, Damysus, and HotStuff. Chained-Raftel and Achilles provide the lowest latencies across much of their throughput range. Latency generally rises as each protocol approaches saturation. The experiment uses a fault threshold of 8 and batches of 400 transactions.}

\label{fig:redis}
\end{figure}

\section{Conclusion} \label{sec:conclusion}

This paper revisits BFT consensus under heterogeneous trust. Instead of assuming that either all replicas are TEE-enabled or none are, we consider a partial-TEE setting where only a subset of replicas are equipped with TEEs. We formalize this setting through a universal partial-TEE model and show that heterogeneous trust fundamentally changes quorum formation and fault tolerance.

Our analysis establishes a tight resilience bound and reveals a threshold phenomenon: partial TEE deployment improves protocol efficiency immediately, but increases fault tolerance only after TEE replicas exceed a critical fraction of the system. Based on this insight, we design a dual-quorum construction and TEE-aware fast paths that safely exploit hardware-enforced non-equivocation.
Building on these principles, we design and implement \sysname and \csysname, the first HotStuff-style protocols for universal partial-TEE deployments. Evaluation on Intel SGX shows that they significantly outperform classical TEE-free BFT protocols and approach the performance of fully TEE-assisted designs.

\begin{acks}
    This work is supported in part by the NSERC Discovery Grant RGPIN-2023-04962; in part by the Singapore Ministry of Education, under Academic Research Fund Tier 2 grant MOE-T2EP20222-0017; in part by NSFC under Grants 62302204; in part by Fundamental and Interdisciplinary Disciplines Breakthrough Plan of the Ministry of Education of China under Grants JYB2025XDXM114; in part by the Research Grants Council of HK under Grants CityU 11219025, 11219524, R1012-21, RFS2122-1S04, CRS\_HKUST601/24, and C6015-23G; and in part by the Guangdong and Hong Kong Universities "1+1+1" Joint Research Collaboration Scheme.
\end{acks}

\normalem

\bibliographystyle{ACM-Reference-Format}
\bibliography{bib}

\clearpage

\end{document}